\documentclass[aps,pra, a4paper,twocolumn,11pt,accepted=2025-01-13]{quantumarticle}
\pdfoutput=1
\usepackage[utf8]{inputenc}
\usepackage[english]{babel}
\usepackage[T1]{fontenc}
\usepackage{amsmath}
\usepackage{hyperref}
\usepackage[numbers,sort&compress]{natbib}

\usepackage{ezkz}
\newcommand\nmberthis{\addtocounter{equation}{1}\tag{\theequation}}
\usepackage{ragged2e}
\usepackage{amssymb}
\usepackage{mathtools}

\usepackage{microtype}

\usepackage{tikz}
\usepackage{lipsum}

\begin{document}

\title{Efficient Hamiltonian encoding algorithms for extracting quantum control mechanism as interfering pathway amplitudes in the Dyson series}

\author{Erez Abrams}
\orcid{0009-0009-7971-9571}
\email{erezm@mit.edu}
\affiliation{Princeton University}
\affiliation{Massachusetts Institute of Technology}
\author{Michael Kasprzak}
\orcid{0009-0008-8304-6798}
\email{mk7592@princeton.edu}
\author{Gaurav Bhole}
\orcid{0000-0002-8300-9822}
\author{Tak‐San Ho}
\orcid{0000-0003-1286-0289}
\author{Herschel Rabitz}
\orcid{0000-0002-4433-6142}
\affiliation{Princeton University}

\maketitle

\begin{abstract}
  Hamiltonian encoding is a methodology for revealing the mechanism behind the dynamics governing controlled quantum systems. In this paper, following Mitra and Rabitz {[Phys. Rev. A 67, 033407 (2003)]}, we define mechanism via pathways of eigenstates that describe the evolution of the system, where each pathway is associated with a complex-valued amplitude corresponding to a term in the Dyson series. The evolution of the system is determined by the constructive and destructive interference of these pathway amplitudes. Pathways with similar attributes can be grouped together into pathway classes. The amplitudes of pathway classes are computed by modulating the Hamiltonian matrix elements and decoding the subsequent evolution of the system rather than by direct computation of the individual terms in the Dyson series. The original implementation of Hamiltonian encoding was computationally intensive and became prohibitively expensive in large quantum systems. This paper presents two new encoding algorithms that calculate the amplitudes of pathway classes by using techniques from graph theory and algebraic topology to exploit patterns in the set of allowed transitions, greatly reducing the number of matrix elements that need to be modulated. These new algorithms provide an exponential decrease in both computation time and memory utilization with respect to the Hilbert space dimension of the system. To demonstrate the use of these techniques, they are applied to two illustrative state-to-state transition problems. 
\end{abstract}

\section{Introduction}

Quantum control generally operates by using an external field, such as a shaped laser pulse or magnetic field, to manipulate the dynamics of a quantum system. Quantum control continues to be studied extensively in numerical simulations \cite{OC_the_1, OC_the_2, OC_the_3, OC_the_4, OC_the_5} and in the laboratory \cite{OC_exp_1, OC_exp_2, OC_exp_3}. It is common for control fields to be tailored such that the system will start at an initial state $\ket{a}$, evolve continuously through superposition states, and eventually arrive at its final state $\ket{b}$. Control fields may also be tailored to manipulate more general physical observables or create particular unitary transformations. This paper will focus on the basic case of state-to-state transitions, but in all circumstances the control induces a particular dynamical mechanism. 

In order to keep a well-defined focus, this paper will consider extracting the mechanism of controlled closed quantum systems with a finite number of discrete eigenstates. A common mechanistic diagnostic is the creation of state population plots with respect to time; in some favorable cases, mechanistic insights may be garnered from such plots, but it is generally difficult to reach a quantitative understanding of the mechanism behind a system's evolution without the use of additional analytical tools. Another common procedure is to examine the control field in a chosen function representation (e.g.~in time, frequency, time-frequency, etc.), and again this procedure can be helpful in some favorable cases. In order to establish a systematic, quantitative, and information-rich mechanism extraction procedure, Mitra and Rabitz proposed revisiting the foundations of quantum control and extracting the constructive and destructive interference of the amplitudes responsible for reaching the objective. In particular, they introduced Hamiltonian encoding techniques which decompose the evolution of the system into pathways of eigenstates by associating them with the individual components of the resultant unitary transformation's Dyson series expansion \cite{abhra_1, abhra_2, abhra_3, abhra_4}. This paper substantially improves the computational efficiency of these Hamiltonian encoding techniques, enabling future mechanism analysis studies of controlled dynamics for more complex quantum systems.

Following \cite{abhra_1}, this paper defines mechanism in terms of pathways corresponding to terms in the Dyson series expansion of the time-evolution operator, where each term defines a complex-valued pathway amplitude representing the contribution of its associated pathway to the overall mechanism. Several ways to classify similar pathways were defined in the prior work \cite{abhra_1} and will be given specific attention here. Sufficiently similar pathways can be grouped together into pathway classes, where the amplitude of a pathway class is the sum of the amplitudes of the individual pathways in the class. Hamiltonian encoding extracts pathway class amplitudes by modulating the Hamiltonian and decoding the subsequent evolution, which is much more efficient than numerically evaluating the amplitudes of each term in the series. This paper will show that the computational cost of Hamiltonian encoding can be significantly reduced by exploiting topological properties of the graph representing the system's allowed transitions, enabling the fast extraction of pathway class amplitudes by eliminating redundant information from the pathway class. We present two new algorithms that utilize these properties: \textit{Optimal Hermitian Pathway Encoding} (OHPE) and \textit{Non-Hermitian Pathway Encoding} (NHPE) each provide an exponential improvement over the original methods \cite{abhra_1} for encoding Hermitian pathway classes and non-Hermitian pathway classes, respectively (Hermitian and non-Hermitian pathway classes are defined in Section \ref{ssec:HNHPE}). In particular, OHPE is proven to be the optimal method of encoding Hermitian pathway classes with respect to both computation time and memory space complexity. As such, this paper provides substantial algorithmic advances in Hamiltonian encoding mechanism analysis; Section \ref{sec:numerical_illustration} demonstrates these improvements numerically. Follow-up papers will explore particular physical applications of these algorithms, e.g.~examining the mechanism of quantum unitary transformations serving as qubit gates.

To provide necessary background, a summary of previous work on Hamiltonian encoding is presented in Section \ref{sec:pathways_and_encoding}. The new encoding schemes, OHPE and NHPE, are presented in Section \ref{sec:OHPE_NHPE}. Numerical illustrations of these algorithms are provided in Section \ref{sec:numerical_illustration}. Concluding remarks regarding future applications of these techniques are presented in Section \ref{sec:conclusion}. Appendices \ref{apx:motiv_exam}, \ref{apx:proof}, and \ref{apx:cycles} provide a motivating example for OHPE, a formal proof of its correctness and optimality, and a novel reinterpretation of its encoding, respectively. Appendix \ref{apx:NHPE} provides an informal proof of correctness for NHPE, and finally Appendix \ref{apx:algo} details supplementary algorithms that aid in the implementation of Hamiltonian encoding.

\section{Pathways and Hamiltonian Encoding}\label{sec:pathways_and_encoding}
This section chiefly reviews previous work as background on analyzing mechanistic pathways extracted from the Dyson series through Hamiltonian encoding \cite{abhra_1, koswara}. Hamiltonian encoding is carried out by modulating the Hamiltonian in an additional time-like parameter $s$, then decoding the resulting evolution of the modulated system. In this algorithm, mechanism is defined as resulting from the constructive and destructive interference of the pathway amplitudes constituting the Dyson series. A description of Hamiltonian encoding via complex exponentials is followed by a discussion of Hermitian and non-Hermitian encoding in more detail. This section also includes some expansion on previous works, especially with regards to the development of a more robust formalism for Hamiltonian encoding (Sections \ref{sssec:domains} and \ref{sssec:choosing_b}).

\subsection{Pathways in the Dyson Series}\label{ssec:pdef}
Quantum systems with Hamiltonians of the form $H = H_0 - \mu \varepsilon(t)$ are common in control scenarios, where $H_0$ is the field free Hamiltonian, $\mu$ is the dipole operator, and $\varepsilon(t)$ is the control field. The eigenvalues and eigenvectors of $H_0$ are denoted $E_i$ and $\ket{i}$, respectively, satisfying $H_0 \ket{i} = E_i \ket{i}$ for $i= 1,2,\dots,d$ where $d$ is the dimension (i.e.~the number of energy levels) of the quantum system. The unitary time-evolution operator $U(t)$ in the interaction picture is governed by the time-dependent Schr\"odinger equation:
\begin{equation}
    i\hbar \dv{t}U(t) = V_I(t) U(t), \label{eq:schro_eq}
\end{equation}
where $V_I(t) \equiv  -\exp(iH_0t/\hbar) \mu\varepsilon(t) \exp(-iH_0t/\hbar)$. The Dyson series expansion of $U(t)$ is given as follows:
\begin{align*}\label{eq:dyson_full}
     U\hspace{-5pt}&\hspace{5pt}(t) = 1 + \qty(\frac{-i}{\hbar})\!\int_{0}^{t} V_I(t_1) \dd{t_1}\\  \notag
     &+ \qty(\frac{-i}{\hbar})^2\!\int_{0}^{t}\! \int_{0}^{t_2} V_I(t_2)  V_I(t_1) \dd{t_1}\!\, \dd{t_2}\\ \notag
     &+ \qty(\frac{-i}{\hbar})^3\!\int_{0}^{t}\! \int_{0}^{t_3}\!\!\! \int_{0}^{t_2}\!\, V_I(t_3)  V_I(t_2)  V_I(t_1)\!\hspace{0.07em} \dd{t_1}\! \dd{t_2}\!\hspace{.07em}\dd{t_3}\\ \notag
     &+\ \cdots. \nmberthis
\end{align*}
Normally the Dyson expansion is introduced as a formal solution of Eq.~\eqref{eq:schro_eq}, and in the present context we are only using it to exract control mechanism information; the numerical solution for $U(t)$ alone would likely follow other traditional numerical methods. Hamiltonian encoding techniques provide practical means for numerically determining the terms in this expansion, sidestepping the need to evaluate the integrals directly. Importantly, a definition of control mechanism expressed in terms of pathways and their amplitudes arises from the physical interpretation of this expansion. A common quantum control problem involves a population transfer from a starting state $\ket a$ to a final state $\ket b$ at final time $T$. Repeatedly inserting the completeness relation $\sum_{i=1}^d \dyad{i} = 1$ into Eq.~\eqref{eq:dyson_full} and defining $v_{ji}(t) \equiv - \frac i \hbar \mel{j}{V_I(t)}{i} =  \frac i \hbar \mel{j}{\mu}{i} e^{i (E_j - E_i) t / \hbar} \varepsilon(t)$, the transition amplitude $U_{ba} \equiv \mel{b}{U(T)}{a}$ can be expressed as:
\begin{equation} \label{eq:dyson_elements}
    U_{ba} = \sum_{n=0}^\infty \sum_{l_{n-1}=1}^d\cdots\sum_{l_1=1}^d U_{ba}^{n(l_1,\dots,l_{n-1})}
\end{equation}
where
\begin{align} \label{eq:pathway_amps}
        U^{n(l_1,\dots,l_{n-1})}_{ba}\hspace{30pt}& \nonumber \\ 
        \equiv \int_0^T\int_0^{t_n}\cdots\int_0^{t_2} &v_{bl_{n-1}}(t_n)v_{l_{n-1}l_{n-2}}(t_{n-1})\cdots \nonumber \\ 
        & v_{l_1a}(t_1) \dd{t_1}\cdots\dd{t_{n-1}}\dd{t_n}.
\end{align}
The index $n$ in Eqs.~\eqref{eq:dyson_elements} and \eqref{eq:pathway_amps} denotes the \textit{order} of the term, i.e.~the number of transitions induced by the interaction operators $V_I$. The parentheses $(l_1,\dots,l_{n-1})$ denote a list of the intermediate states in the $n$th term of the Dyson series. A \textit{pathway} between $\ket{a}$ and $\ket{b}$ is a sequence of $n$ transitions \mbox{$\ket{a} \!\to\! \ket{l_1} \!\to\! \cdots \!\to\! \ket{l_{n-1}} \!\to\! \ket{b}$} with \mbox{$n-1$} intermediate states $\ket{l_i}$.  Figure \ref{fig:3l_ex_paths} displays several examples of pathways in a model three-level system. Eq.~\eqref{eq:dyson_elements} decomposes the Dyson series into terms that each correspond to an individual pathway; the complex value $U^{n(l_1,\dots,l_{n-1})}_{ba}$ is the \textit{path\-way amplitude} of \mbox{$\ket{a} \!\to\! \ket{l_1} \!\to\! \cdots \!\to\! \ket{l_{n-1}} \!\to\! \ket{b}$.} Pathways with large amplitudes have a large contribution to the controlled dynamics, and these complex amplitudes constructively and destructively interfere to drive the evolution of the system. The goal of mechanism analysis in this framework is to examine the individual terms of Eq.~\eqref{eq:dyson_elements} and determine those which contribute significantly to the interference between pathways.

It is often beneficial to group similar pathways into a \textit{pathway class} (whose amplitude is the sum of all individual pathway amplitudes within the class) when these pathways differ only by features which are deemed unimportant when seeking a coarse-grained assessment of the mechanism. Two notable examples of such features are \textit{time-sequencing} and \textit{backtracking}. Time-sequencing describes the sequence in which transitions occur in a pathway, so \mbox{$\ket 1 \!\to\! \ket 2 \!\to\! \ket 1 \!\to\! \ket3 \!\to\! \ket1 \!\to\! \ket2$} and \mbox{$\ket1 \!\to\! \ket3 \!\to\! \ket1 \!\to\! \ket2 \!\to\! \ket1 \!\to\! \ket2$} differ only in the time-sequencing of their transitions. Backtracking is when a pathway includes a transition from state $\ket i$ to $\ket j$ and later a transition from $\ket j$ to $\ket i$; this includes Rabi flopping, where a pathway includes a transition from $\ket i$ to $\ket j$ and immediately thereafter a transition from $\ket j$ back to $\ket i$. Often, the details of how many times a pathway backtracks between two states are unimportant in extracting the essence of the mechanism. Furthermore, grouping together pathways that differ only by time-sequencing and backtracking greatly simplifies the extracted mechanism's major details and their assessment. To this end, it is particularly useful to define a \textit{Hermitian pathway class} as a set of pathways whose members differ only by time-sequencing and backtracking (a more rigorous definition is provided in Section \ref{sssec:H}). We will later encounter encodings which distinguish backtracking, but we will not consider encodings which distinguish time-sequencing in this paper; if time-sequencing information is desired, more comprehensive encodings can be utilized \cite{abhra_1}. 

The lowest-order (shortest) pathway in a pathway class is called its \textit{representative pathway}. There may be multiple pathways with the same minimal order in a pathway class, in which case one of these is arbitrarily chosen to represent the pathway class. We denote Hermitian pathway classes by their representative pathway in brackets with a superscript H, and the pathway amplitudes of Hermitian classes are denoted with an H in their superscript as $U_{ba}^{n(l_1,\dots,l_{n-1})\text{H}}$.\footnote{In previous works, Hermitian pathway classes were denoted with asterisks rather than H superscripts.} For example, \mbox{$\ket{1} \!\to\! \ket{2} \!\to\! \ket{3}$} and \mbox{$\ket{1} \!\to\! \ket{2} \!\to\! \ket{1} \!\to\! \ket{2} \!\to\! \ket{3}$} are both members of the Hermitian pathway class \mbox{$[\ket{1} \!\to\! \ket{2} \!\to\! \ket{3}]^\text{H}$}, but the latter pathway backtracks the \mbox{$\ket1\!\to\!\ket2$} transition. Another example of a pathway containing backtracking is given in green arrows in Figure \ref{fig:3l_ex_paths}.

\begin{figure}[t]\captionsetup{font=small}
    \centering
    \includegraphics[width=\linewidth]{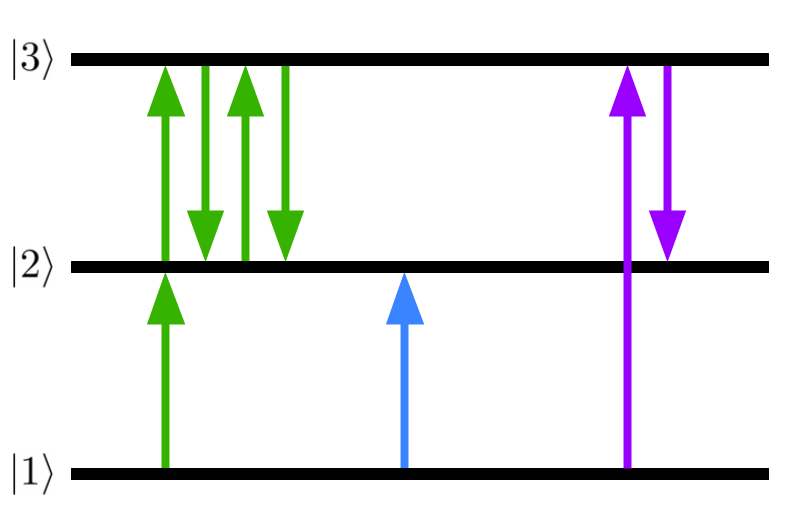}
    \caption{\justifying An illustrative three-level quantum system with three pathways shown between $\ket{1}$ and $\ket{2}$. The first (green) pathway is \mbox{$\ket 1 \!\to\! \ket 2 \!\to\! \ket 3 \!\to\! \ket 2 \!\to\! \ket 3 \!\to\! \ket 2$} which has amplitude $U^{5(2,3,2,3)}_{21}$ and makes two backtracking transitions between $\ket 2$ and $\ket3$. The second (blue) pathway is the direct transition \mbox{$\ket{1} \!\to\! \ket{2}$} and has amplitude $U^{1()}_{21}$. Both of these pathways are members of the Hermitian pathway class \mbox{$[\ket 1 \!\to\! \ket 2]^\text{H}$.} The third (purple) pathway is \mbox{$\ket 1 \!\to\! \ket 3 \!\to\! \ket 2$}, which has amplitude $U^{3(3)}_{21}$ and is the shortest pathway in the Hermitian pathway class \mbox{$[\ket 1 \!\to\! \ket 3 \!\to\! \ket 2]^\text{H}$.}}
    \label{fig:3l_ex_paths}
\end{figure}

\subsection{Fourier Hamiltonian Encoding}\label{ssec:fourier}
Each $n$th-order pathway amplitude $U_{ba}^{n(l_1,\dots,l_{n-1})}$ in Eq.~\eqref{eq:pathway_amps} may, in principle, be calculated directly via numerical integration, but as $n$ increases, the number of possible pathways grows exponentially due to the increasing number of intermediate states $(l_1,\dots,l_{n-1})$, so evaluating the integrals for every pathway individually becomes computationally prohibitive. Instead of directly computing these integrals, Hamiltonian encoding modulates the Hamiltonian with functions of an additional time-like parameter $s$ and decodes the evolution of the modulated system to extract the pathway amplitudes $U_{ba}^{n(l_1,\dots,l_{n-1})}$.

The most common implementation of Hamiltonian encoding is \textit{Fourier encoding}, in which each matrix element of the dipole is multiplied by a complex exponential in $s$ as follows:
\begin{equation} \label{eq:general_mod_mu}
    \mu_{ji} \rightarrow\mu_{ji}(s)\equiv e^{i\gamma_{ji}s}\mu_{ji},
\end{equation}
where each matrix element of $\mu$ is assigned a frequency $\gamma_{ji}$. Importantly, the modulation frequencies $\gamma_{ji}$ are strictly chosen for encoding purposes, and they have no relation to the transition frequencies $(E_j-E_i)/\hbar$ of the quantum system. The interaction Hamiltonian is thus modulated via the relation:
\begin{equation} \label{eq:general_mod}
    v_{ji}(t) \rightarrow v_{ji}(t; s) \equiv  e^{i\gamma_{ji}s}v_{ji}(t)
\end{equation}
which yields the modulated Schr\"odinger equation:
\begin{equation}\label{eq:mod_schr}
    \dv{U\!\,(t; s)}{t} =\! \pmqty{e^{i\gamma_{11}s}v_{11}\!\,(t) & \!\!\!\cdots\!\!\! & e^{i\gamma_{1d}s}v_{1d}(t) \\ \vdots & \!\!\!\ddots\!\!\! & \vdots \\ e^{i\gamma_{d1}s}v_{d1}\!\,(t) & \!\!\!\cdots\!\!\! & e^{i\gamma_{dd}s}v_{dd}(t)}\!U\!\,(t; s).
\end{equation}
As a result, the encoded amplitude for the pathway \mbox{$\ket a \!\to\! \ket{l_1} \!\to\! \cdots \!\to\! \ket{l_{n-1}} \!\to\! \ket b$} can be written as:
\begin{align*} \label{eq:modulated_pathway_amp}
     &\hspace{3.7pt}U^{n(l_1,\dots,l_{n-1})}_{ba}(s)\\ \notag
        &\equiv \int_0^T\int_0^{t_n}\cdots\int_0^{t_2} v_{bl_{n-1}}(t_n)e^{i\gamma_{bl_{n-1}}s}\cdots\notag \\ 
        & \hphantom{\int_0^T\int_0^{t_n}\cdots\int_0^{t_2}}\hspace{15pt}v_{l_1a}(t_1)e^{i\gamma_{l_1a}s}\dd{t_1}\cdots\dd{t_{n-1}}\dd{t_n}  \\ 
        &= U^{n(l_1,\dots,l_{n-1})}_{ba}\, e^{i\gamma^{n(l_1,\dots,l_{n-1})}_{ba}s} \nmberthis
\end{align*}
where 
\begin{equation} \label{eq:gamma_pathway}
    \gamma^{{n(l_1,\dots,l_{n-1})}}_{ba} \equiv \gamma_{bl_{n-1}} + \dots + \gamma_{l_1a},
\end{equation}
indicating that each pathway now is associated with a pathway frequency $\gamma^{{n(l_1,\dots,l_{n-1})}}_{ba}$. As a result, the modulated $\ket{a}$ to $\ket{b}$ probability amplitude can be written as follows:
\begin{equation} \label{eq:Abhra_linear_eq}
    U_{ba}(T;s) =\! \rlap{\phantom{\Big)}}\smash{\sum_{{\text{\shortstack{pathways\\$(l_1,\dots,l_{n-1})$}}}} }\!U^{n(l_1,\dots,l_{n-1})}_{ba}\! \cdot\! e^{i\gamma^{{n(l_1,\dots,l_{n-1})}}_{ba}s} \vspace{13pt}
\end{equation}
which, given sufficiently many samples in $s$ and distinct frequencies $\gamma^{{n(l_1,\dots,l_{n-1})}}_{ba}$, can be efficiently solved for the pathway amplitudes $U^{n(l_1,\dots,l_{n-1})}_{ba}$ via a fast Fourier transform. This process is laid out diagrammatically in Figure \ref{fig:flowchart}.

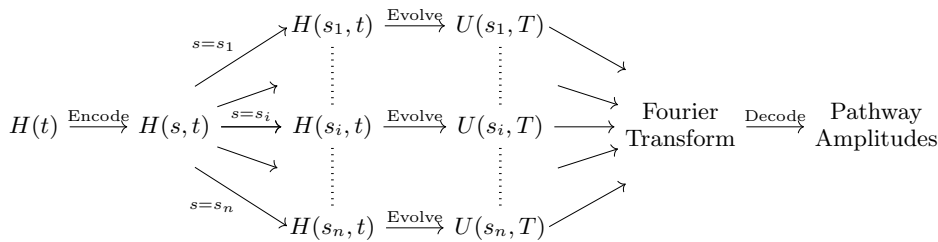
\begin{figure*}
    \centering
    \begin{tikzcd}
                                 &                                                                                & {H(s_1, t)} \arrow[r, "\textrm{Evolve}"] \arrow[d, no head, dotted, line width=1pt] & {U(s_1, T)} \arrow[d, no head, dotted, line width=1pt] \arrow[rd, shift left=4] &               &                           \\
H(t) \arrow[r, "\textrm{Encode}"] & {H(s,t)} \arrow[rd, shift right=4, "s=s_n"'] \arrow[ru, "s=s_1", shift left=4] \arrow[r, "s=s_i"] \arrow[r, shift left=2, end anchor={[yshift=.9ex]north west}] \arrow[r, shift right=2, end anchor={[yshift=-.9ex]south west}] \arrow[r] & {H(s_i, t)} \arrow[r, "\textrm{Evolve}"] \arrow[d, no head, dotted, line width=1pt] & {U(s_i, T)} \arrow[d, no head, dotted, line width=1pt] \arrow[r] \arrow[r, shift left=2, start anchor={[yshift=.9ex]north east}] \arrow[r, shift right=2, start anchor={[yshift=-.9ex]south east}]  & \makecell[c]{\textrm{Fourier} \\ \textrm{Transform}} \arrow[r, "\textrm{Decode}"] & \makecell[c]{\textrm{Pathway} \\ \textrm{Amplitudes}} \\
                                 &                                                                                & {H(s_n, t)} \arrow[r, "\text{Evolve}"]                            & {U(s_n, T)} \arrow[ru, shift right=4]                            &               &                          
\end{tikzcd}
    \captionsetup{font=small}\caption{\justifying A flowchart showing an outline of the Hamiltonian encoding process. The (time-dependent) Hamiltonian is modulated in an additional time-like parameter $s$. At many values of $s$, the modulated Schr\"odinger equation is integrated over time to capture the modulated system's evolution. The amplitudes of a large number of pathway classes are simultaneously extracted by decoding the evolution of these modulated systems via a fast Fourier transform.}
    \label{fig:flowchart}
\end{figure*}

\subsubsection{Pathway Classes}
Here we formally define the concept of a pathway class, which is fundamental to the encoding process. The specific choice of frequencies used in Fourier encoding (Eqs.~\eqref{eq:dyson_elements} and \eqref{eq:pathway_amps}) is called the \textit{encoding scheme}, or encoding for short. Under any given encoding, it is likely that multiple pathways share the same pathway frequency $\gamma^{n(l_1,\dots,l_{n-1})}_{ba}$. As a result, Eq.~\eqref{eq:Abhra_linear_eq} can be recast as the following:
\begin{align}
        U_{ba}(T;s)=&\! \sum_{{\text{\shortstack{pathways\\$(l_1,\dots,l_{n-1})$}}}} \!U^{n(l_1,\dots,l_{n-1})}_{ba} \cdot e^{i\gamma^{{n(l_1,\dots,l_{n-1})}}_{ba}s} \notag\\
        =&\ \sum_{\text{frequencies $\gamma$}} e^{i \gamma s } U_{ba}^{\gamma}\label{eq:actual_abhra_linear_eq}
    \intertext{where}
U_{ba}^\gamma \equiv& \sum _{{ \text{\shortstack{pathways\\$(l_1,\dots,l_{n-1})$ with\\[-2pt]frequency $\gamma$}}}} U_{ba}^{{n(l_1,\dots,l_{n-1})}}\label{eq:pathway_classes}\vspace{-15pt}
\end{align}

Pathways that share the same frequency $\gamma$ in an encoding are said to be in the same \textit{pathway class} under that encoding, and the amplitude $U_{ba}^\gamma$ of the pathway class is the sum of the individual pathway amplitudes in the class. Because Hamiltonian encoding yields pathway class amplitudes, it is critical that pathway classes are formed by grouping together pathways according to some physical property to extract useful mechanistic information. Section \ref{ssec:HNHPE} will explore non-Hermitian encoding, where pathways that differ only by the time-sequencing of transitions are in the same pathway class (i.e.~are not distinguished from each other), and Hermitian encoding, where all pathways that differ by backtracking and time-sequencing are in the same pathway class.

\subsubsection{Encoding Domains and Base \texorpdfstring{$B$}{B} Encoding}\label{sssec:domains}
Pathways may contain repeated transitions, so it is often useful to decompose a pathway frequency into individual modulation frequencies (basis frequencies). These basis frequencies will be labeled by integer indices $k \geq 1$ so that if the $k$th basis frequency is the modulation frequency of \mbox{$\ket i\! \to \!\ket j$} then $\gamma_k\equiv\gamma_{ji}$. Since pathway frequencies $\gamma^{n(l_1,\dots,l_{n-1})}_{ba}$ arise from a sum of individual basis frequencies (Eq.~\eqref{eq:gamma_pathway}), they can in general be decomposed as follows:
\begin{equation} \label{eq:basis_freq}
    \gamma^{n(l_1,\dots,l_{n-1})}_{ba} = \sum_k n_k\gamma_k.
\end{equation}
It is clear in this expansion that the $k$th transition is traversed $n_k$ times. This means that if the modulation frequencies $\gamma_k$ are chosen so that the coefficients $n_k$ can be uniquely reconstructed from $\gamma^{n(l_1,\dots,l_{n-1})}_{ba}$, the pathway class frequency can completely describe the number of times a pathway traverses each transition. If two pathways have the same pathway frequency, they will be indistinguishable in the encoding and their amplitudes will be added together in the Fourier transform. To distinguish $S$ different pathway class frequencies in Eq.~\eqref{eq:actual_abhra_linear_eq}, a Fourier transform needs at least $S$ sample points, so the runtime of an encoding is $O(S)$ numerical integrations of Eq.~\eqref{eq:mod_schr}. 

Although the Dyson series contains infinitely many terms (and therefore infinitely many pathways), in practice, the series converges and only a finite number of those pathways and pathway classes have a large enough amplitude to contribute significantly to the mechanism. To make this notion precise, we define some small amplitude magnitude threshold $\epsilon$; if the amplitude of a pathway class is larger in magnitude than $\epsilon$, then we call that pathway class \textit{significant}. To facilitate a meaningful analysis, $\epsilon$ should be set small enough to capture the essence of the mechanism but large enough to not waste computational resources; in Section {\ref{sssec:choosing_b}}, this magnitude threshold will be used to validate choices made in the encoding process.

In order to efficiently extract the mechanism associated with a control, we must make an initial assumption about which pathway classes can be significant (which can be later updated as described in Section \ref{sssec:choosing_b}). We call the set of pathway classes that an encoding distinguishes its \textit{encoding domain}. Every pathway class in an encoding domain is assigned a unique frequency, so with every additional pathway class in an encoding domain, an additional $s$-point is needed to keep Eq.~\eqref{eq:actual_abhra_linear_eq} completely determined. Increasing the number of pathway classes in a domain increases computation time while enabling a more detailed view of the underlying mechanism. Consequently, choosing a good encoding corresponds to choosing the smallest possible encoding domain that includes all significant pathway classes (as determined by the amplitude threshold $\epsilon$).

Once an encoding domain is chosen, modulation frequencies must be chosen so that all pathway classes in the domain are distinguishable, i.e.~no two pathway classes in the domain share the same pathway class frequency. For example, if the encoding domain is chosen to be, for some integer $B$, the set of all pathways that use each transition fewer than $B$ times, then each of the ${r}$ encoded transitions in $\mu$ may be given the frequency  $\gamma_k = {B}^{k-1}\gamma_0$ where $k=1, \dots, {r}$ and $\gamma_0$ is some chosen reference frequency. This is a natural way to encode pathway classes in the aforementioned domain \cite{koswara}. We call this choice of frequencies \textit{base $B$ encoding} and develop it further in Section \ref{ssec:HNHPE}.

Base $B$ encoding yields ${B}^{r}$ possible pathway frequencies in the encoding domain. Every pathway in the domain has a unique frequency up to time-sequencing, and in total, we require ${B}^{r}$ sample points in $s$. However, fast Fourier transforms are often performed on sets of $2^N$ evenly spaced sample points for some $N$ \cite{ct-fft}, i.e.~samples at $s = 2 \pi k/ 2^N$ for all $0 \leq k < 2^N$. In addition, $\gamma_0$ must be chosen such that $ B^r\gamma_0 \leq 2\pi$ so that all pathway class frequencies are less than $2\pi$. Both of these conditions can be satisfied by choosing $N = {\lceil  r\log_2 B\rceil}$ so that $2^N$ is the largest power of 2 no less than $B^r$:
\begin{equation}
    \gamma_0 = \frac{2\pi}{2^{\lceil  r\log_2 B\rceil}}.
\end{equation}
Therefore, encoding ${r}$ transitions with base $B$ takes $2^{\lceil  r\log_2 B\rceil}=O(B^r)$ numerical integrations of Eq.~\eqref{eq:mod_schr}. The process of appropriately choosing $B$ to minimize computational cost while preserving all essential mechanistic information is detailed in Section {\ref{sssec:choosing_b}}.

\subsection{Hermitian and Non-Hermitian Hamiltonian Encoding}\label{ssec:HNHPE}
The optimizations described in Section \ref{sec:OHPE_NHPE} consider Hermitian and non-Hermitian pathway classes. Sections \ref{sssec:H} and \ref{sssec:NH} will describe Hermitian and non-Hermitian pathway classes respectively in more detail than previous works \cite{abhra_1}, and Section \ref{sssec:HNHP} will explore the mechanistic information they each provide. Section \ref{sssec:choosing_b} will detail how $B$ should be chosen in base $B$ implementations of Hermitian and non-Hermitian encoding.
\subsubsection{Hermitian Encoding} \label{sssec:H}
For the modulated Hamiltonian in Eq.~\eqref{eq:mod_schr} to remain Hermitian, we use  \textit{Hermitian encoding}:
\begin{equation} \label{eq:hermitian_encoding}
\begin{matrix}
        v_{ji}(t) \rightarrow v_{ji}(t) e^{i\gamma_{ji}s}\\
        \gamma_{ij} = - \gamma_{ji}.
\end{matrix}
\end{equation}
Inspecting Eq.~\eqref{eq:gamma_pathway} with this frequency choice shows that any backtracking will cancel out in the sum for $\gamma^{{n(l_1,\dots,l_{n-1})}}_{ba}$, so the coefficient $n_k$ of $\gamma_k \equiv \gamma_{ji}$ in Eq.~\eqref{eq:basis_freq} describes the number of \mbox{$\ket i \!\to\! \ket j$} (forward) transitions minus the number of \mbox{$\ket j \!\to\! \ket i$} (backward) transitions of a pathway in that class (i.e.~the \textit{net} number of \mbox{$\ket i \!\to\! \ket j$} transitions). Therefore, all pathways that differ only by backtracking and time-sequencing share the same $\gamma^{{n(l_1,\dots,l_{n-1})}}_{ba}$ and they are grouped together under Hermitian encoding. We define the \textit{Hermitian pathway class} of a frequency $\gamma$ to be the set of all pathways that share $\gamma^{n(l_1,\dots,l_{n-1})}_{ba} = \gamma$ under Hermitian encoding (this set was called a ``composite pathway'' in previous works \cite{abhra_1,abhra_2,abhra_3,abhra_4}). We call the frequency of a pathway under a given Hermitian encoding its \textit{Hermitian frequency}. Because $\gamma_{ij} = -\gamma_{ji}$, all pathways that differ only by backtracking share the same Hermitian frequency, so this definition of Hermitian pathway classes is the same as the definition provided in \ref{ssec:pdef}.

The most commonly considered domain in Hermitian encoding is, for some positive integer ${m_0}$, the set of all pathways that use each transition a \textit{net} number of times that is smaller than or equal in magnitude to ${m_0}$ (i.e.~each transition may be used $-{m_0}, -{m_0}+1, \dots, -1, 0, 1, \dots, {m_0}-1,$ or ${m_0}$ times). To distinguish all pathway classes in this domain, we can use a modification of base $B$ encoding. In this encoding, we set $B=2{m_0}+1$ and encode the upper triangle of the Hamiltonian with $\gamma_k = B^{k-1} \gamma_0$, then encode the lower triangle with the negative values of these frequencies to preserve Hermiticity. This is called \textit{Hermitian base $B$ encoding}, and the domain represents the set of all pathways that use each transition a net number of times that is smaller than \mbox{$(B+1)/2$}. For this reason, in Hermitian base $B$ encoding, $B$ is often chosen to be odd so that the above interpretation is valid. Note that because negative frequencies are utilized in this encoding, the Hermitian frequencies extracted from the Fourier transform should be interpreted in the range $[-\pi, \pi)$ rather than $[0, 2\pi)$.

\subsubsection{Non-Hermitian Encoding} \label{sssec:NH}
By design, Hermitian pathway classes do not distinguish pathways that differ by time-sequencing or backtracking. To distinguish pathways that differ by backtracking, we relax the second condition in Eq.~\eqref{eq:hermitian_encoding} with a \textit{non-Hermitian encoding}:
\begin{equation} \label{eq:non_hermitian_encoding}
\begin{matrix}
        v_{ji} \rightarrow v_{ji} e^{i\gamma_{ji}s}\\
        \gamma_{ij} \neq -\gamma_{ji}.
\end{matrix}
\end{equation}
This relaxation means the modulated Hamiltonian (Eq.~\eqref{eq:mod_schr}) is no longer Hermitian. Even though this encoding preserves backtracking information, different pathways will still have the same pathway frequency if they differ by time-sequencing. Mirroring Hermitian pathway classes, we define the \textit{non-Hermitian pathway class} of $\gamma$ to be the set of all pathways that share $\gamma^{{n(l_1,\dots,l_{n-1})}}_{ba} = \gamma$ under non-Hermitian encoding. We denote non-Hermitian pathway classes with a superscript NH after their representative pathway, and the pathway amplitudes of non-Hermitian classes are denoted with an NH in their superscript as $U_{ba}^{n(l_1,\dots,l_{n-1})\text{NH}}$. Note that the pathways of any given non-Hermitian pathway class all contribute to the same Hermitian pathway class.

A convenient domain for non-Hermitian encoding is, for some integer $B$, the set of all pathways that use each transition fewer than $B$ times. To distinguish all pathway classes in this domain, we use the non-Hermitian base $B$ encoding described in Section \ref{sssec:domains}: for each forward or backward transition $k$, let $\gamma_k = B^{k-1} \gamma_0$.

\subsubsection{Comparison of Hermitian and Non-Hermitian Encoding} \label{sssec:HNHP}

The inclusion of backtracking and especially Rabi flopping mechanistic information in non-Hermitian encoding can lead to significant contributions from high-order pathways in which a single transition is traversed a large number of times, so non-Hermitian encoding often requires a larger $B$ than Hermitian encoding to capture all significant pathway classes.  This is seen in the example in Figure \ref{fig:3l_ex_paths}: with Hermitian encoding, these backtracking transitions cancel out, so $B$ can be made smaller, reducing the needed amount of $s$ sample points. Thus, non-Hermitian encoding is more expensive to execute than Hermitian encoding in two ways: $B$ must often be made larger to distinguish pathways with many repeated transitions, and ${r}$ is doubled due to the distinction between the frequencies of forward and backward transitions. As a result, the more detailed mechanistic information provided by non-Hermitian encoding comes at the expense of increased computational effort. Fortunately, despite the lack of backtracking information in Hermitian pathway classes, the information they provide is often sufficient to extract the essence of the mechanism in a system.

The procedures for decomposing $\gamma^{{n(l_1,\dots,l_{n-1})}}_{ba}$ into basis frequencies are distinct for Hermitian and non-Hermitian encodings. Succinctly, decomposing pathway class frequencies from a non-Hermitian base $B$ encoding into basis frequencies is equivalent to writing $\gamma^{{n(l_1,\dots,l_{n-1})}}_{ba}/\gamma_0$ numerically in base $B$ and reading off its digits. Similarly, decomposing pathway class frequencies from a Hermitian base $B$ encoding is equivalent to writing $\gamma^{{n(l_1,\dots,l_{n-1})}}_{ba}/\gamma_0$ numerically in a signed-digit representation such as \textit{balanced base $B$}, which uses the signed {digits} {$\Bqty{-\frac{B-1}{2},-\frac{B-1}{2} + 1,\dots, 0, \dots, \frac{B-1}{2} - 1, \frac{B-1}{2}}$} rather than the unsigned digits $\Bqty{0, \dots, B-1}$, and reading off the (signed) digits in this representation. Only odd bases have a balanced representation for the same reason that $B$ is often chosen to be odd for Hermitian encoding: $B$ must equal $2m_0 + 1$ for some $m_0$ for balanced base $B$ to exist as described above. Algorithms \ref{alg:coef_nherm_freq} and \ref{alg:coef_herm_freq} in Appendix \ref{apx:algo_FD} perform these conversions explicitly.

\subsubsection{Avoiding Aliasing with Self-Validating Base \texorpdfstring{$B$}{B} Encodings}\label{sssec:choosing_b}

Whenever Fourier encoding is performed, a set of modulation frequencies must be chosen, which in the case of base $B$ encoding reduces to a choice of $B$. If $B$ is made too large, many nonsignificant pathway classes will be encoded and excessive computation time will be wasted. If $B$ is made too small, \textit{aliasing} will occur between significant pathway classes. Aliasing is when two significant pathway classes that are physically distinct are indistinguishable in an encoding due to having the same frequency; this occurs only when there exists a significant pathway class outside the encoding domain. As an example, consider a pathway from $\ket 1$ to $\ket 2$ in a 3-level system with only ladder-climbing transitions under base 3 non-Hermitian encoding: $\gamma_{21} = \gamma_0$, $\gamma_{32} = 3\gamma_0$, $\gamma_{12} = 9 \gamma_0$, and $\gamma_{23} = 27\gamma_0$. The domain of this encoding is the set of all pathways that use each transition less than 3 times. If, despite being outside the domain, the pathway \mbox{$\ket1 \!\to\! \ket2 \!\to\! \ket1 \!\to\! \ket2 \!\to\! \ket1 \!\to\! \ket2 \!\to\! \ket1 \!\to\! \ket2$} is significant, we will encounter aliasing: the frequency of this pathway is $4\gamma_{21} + 3\gamma_{12} = 31\gamma_0$, but so is the frequency of the pathway \mbox{$\ket1 \!\to\! \ket2 \!\to\! \ket3 \!\to\! \ket2$} which is $\gamma_{21} + \gamma_{32} + \gamma_{23} = 31\gamma_0$, so the amplitude associated with frequency $31\gamma_0$ will end up being the sum of the two pathways' class amplitudes. We call a base \textit{admissible} if encoding with it does not lead to aliasing between significant pathways; thus, the smallest admissible base $B_{\min}$ is the minimal base $B$ such that all significant pathway classes in the mechanism are in the domain associated with $B$ (for non-Hermitian encoding this means that all pathways use any given transition less than $B$ times, and for Hermitian encoding this means that all pathways use any given transition a net number of times smaller than $\frac{B+1}{2}$). Note that since admissibility is defined in terms of significant pathways,  $B_{\min}$ depends strongly on the choice of $\epsilon$.

In principle, the best $B$ to use would be $B_{\min}$, but it is not possible to determine $B_{\min}$ \textit{a priori}. Instead, we aim to encode with a \textit{self-validating base} \mbox{$B> B_{\min}$}. Any base \mbox{$B > B_{\min}$} is called self-validating because it demonstrates that the domain of $B_{\min}$ distinguishes all significant pathways (i.e.~is admissible) by having nonsignificant amplitudes for all of its extremal pathway classes. If we assume that the set of significant pathway classes is connected,\footnote{We say that two pathway classes are neighbors if they differ by $-1$, 0, or 1 along each basis frequency, and a set of pathway classes is connected if it is path-connected under this topology.} having nonsignificant amplitudes for all extremal pathways for some base $B$ implies that \mbox{$B> B_{\min}$} and thus that the domain of $B$ (which contains the domain of $B_{\min}$) includes all significant pathway classes. To increase confidence that this assumption holds and the chosen base is greater than $B_{\min}$, $B$ may be made larger at the cost of increased computation time, but in practice this assumption usually holds and the optimal choice of $B$ for a given $\epsilon$ is the smallest self-validating base $B_{\min}+1$.

Due to this self-validating behavior, given a choice of $\epsilon$, there are two straightforward approaches to choosing $B$. The first approach starts with a small initial value of $B$ and repeatedly encodes the system with larger and larger $B$ until a self-validating base is reached, at which point the process terminates. The second approach starts with a large initial value of $B$ with the hope that $B$ will be immediately self-validating and only one encoding will be necessary; if $B$ is not immediately self-validating, $B$ is increased until it becomes self-validating (as in the first approach). The first approach is potentially more computationally efficient but requires multiple analyses; this paper utilizes the second approach, which sacrifices some computation time for increased confidence that the connectedness assumption holds.

As stated in Section \ref{sssec:H}, $B$ is usually chosen to be odd in Hermitian encoding so that it equals $2m_0+1$ for some $m_0$. For the examples illustrated in this paper,  $7=2\times 3 + 1$ was a good first guess for $B$ in Hermitian encoding: while it was common for significant Hermitian pathway classes to traverse a transition 2 net times, it was rarer for a significant pathway class to traverse a transition 3 net times, making this choice of $B$ often self-validating. Again, if the initial choice of $B$ is not self-validating, it should be increased as necessary. It is harder to find a good first guess for $B$ in non-Hermitian encoding because in this case $B_{\min}$ is strongly tied to the order of the Dyson series; often, in non-Hermitian encoding, a mix of both approaches to choosing $B$ must be taken. 

\section{Optimizing Hamiltonian Encoding}\label{sec:OHPE_NHPE}
This section describes novel algorithms for optimizing the encoding of pathway class amplitudes. First, the theoretical and practical sources of the original encodings' inefficiency are described. Then, two novel encoding schemes that address these inefficiencies will be proposed. The procedure to optimally perform Hermitian encoding of an arbitrary system (OHPE) is provided in Section \ref{ssec:OHPE}, and Section \ref{ssec:NHPE} adapts this procedure for optimizing non-Hermitian encoding (NHPE).

From a theoretical perspective, the source of the inefficiency of the original encodings is that many pathway classes in their domains did not contain any pathways, so the sum in Eq.~\eqref{eq:pathway_classes} was empty and $U_{ba}^\gamma$ was exactly 0. For example, consider the set of all pathways from $\ket 1$ to $\ket 2$ in a 3-level system with all transitions allowed. Apply base $B=7$ Hermitian encoding (where the domain comprises all pathways that use any given transition a net number of times smaller than 4) and encode the frequencies with $\gamma_{21} = \gamma_0$, $\gamma_{31} = 7\gamma_0$, $\gamma_{32} = 49\gamma_0$, and $\gamma_{ij}=-\gamma_{ji}$. The Hermitian frequency $\gamma = 49 \gamma_0$ corresponds to the pathway class composed of pathways that use the \mbox{$\ket 2 \!\to\! \ket 3$} transition 1 net time and every other transition 0 net times; inspecting these conditions reveals that no self-consistent pathway from $\ket 1$ to $\ket 2$ can satisfy them, so the pathway class is empty. Applying the same procedure to a non-Hermitian encoding scheme will reveal that non-Hermitian encodings also have many empty pathway classes. Each pathway class in an encoding has its own unique frequency and each frequency requires an $s$-point, and there are far more empty classes than nonempty classes in the original encodings (see Appendix \ref{apx:proof}); such cases lead to a disproportionately large number of unnecessary integrations of Eq.~\eqref{eq:mod_schr}. Optimized encoding domains should aim to include as few empty pathway classes as possible, reducing the required number of integrations of Eq.~\eqref{eq:mod_schr}. 

From a practical perspective, the source of the inefficiency of the original encodings was the inclusion of redundant information. Appendix \ref{apx:motiv_exam} demonstrates that modulating all transitions in the Hamiltonian is unnecessary and all pathway classes can be reconstructed from the information of a reduced set of transitions. Encoding fewer transitions means that fewer independent basis frequencies are needed, thereby making the encoding less expensive to execute.  

This paper introduces a Hermitian encoding scheme that does not distinguish any empty pathway classes while yielding the same mechanistic information as the original encoding. This scheme is also adapted for non-Hermitian encoding. The new procedure for Hermitian encoding is called Optimal Hermitian Pathway Encoding (OHPE) (Section \ref{ssec:OHPE}). Appendix \ref{apx:proof} proves that OHPE encodes a minimal subset of the system's transitions such that the set of allowed frequencies is in bijection with the set of nonempty pathway classes in the domain. The newly-introduced procedure for non-Hermitian encoding is called Non-Hermitian Pathway Encoding (NHPE) (Section \ref{ssec:NHPE}). Unlike OHPE, this procedure may not be optimal (some empty pathway classes may still be encoded), but it still provides a similarly significant improvement in computation time.

Each optimization presented in this section relies on a construct that we call the \textit{Hamiltonian graph}, which is built by constructing one vertex for each eigenstate of the Hamiltonian and one edge for every allowed transition (the nonzero elements of the dipole matrix). For NHPE, the Hamiltonian graph will be slightly modified into a \textit{directed Hamiltonian graph} so that forwards and backwards transitions are manifestly distinct: instead of having one undirected edge $e_{ji}$ for each transition \mbox{$\ket i \!\to\! \ket j$} with $i<j$, the directed Hamiltonian graph has one forward arc (directed edge) $e_{ji}$ and one backward arc $e_{ij}$ for each transition. The shape of this graph is the same as the undirected Hamiltonian graph, but it has twice as many edges (one for each direction).

We remark that the mathematical basis of OHPE (and NHPE) draws on graph theory concepts, including considering a decomposition of pathways in terms of cycles on the Hamiltonian graph (see Appendix \ref{apx:proof}). Even though the generalization and proof of this optimization both rely heavily on the notion of algebraic topological cycles, expressing pathways in terms of cycles rather than transitions requires a substantial shift in perspective, so the bulk of this paper will not discuss cycles. A detailed description of pathways in terms of cycles is presented in a motivating example and along with the proof of the OHPE's optimality in Appendices \ref{apx:motiv_exam} and \ref{apx:proof}, respectively, with additional discussion in Appendix \ref{apx:cycles}.

\subsection{Optimal Hermitian Pathway Encoding}\label{ssec:OHPE}

Here we will describe which matrix elements of the Hamiltonian should be modulated in OHPE for the optimal extraction of Hermitian pathway class amplitudes. With the modulated Hamiltonian, Eq.~\eqref{eq:mod_schr} is solved at a sufficient number of $s$-points to perform a fast Fourier transform and the amplitudes of all pathway class frequencies in the encoding domain are calculated. The Hamiltonian is modulated as follows:
\begin{enumerate}
    \item Generate the Hamiltonian graph with states as vertices and allowed transitions as edges.
    \item Generate a spanning tree of the Hamiltonian graph. A spanning tree of a graph is a subgraph that has no cycles and includes every vertex of the graph. 
    \item Encode only the transitions corresponding to edges \textit{not} included in the spanning tree such that $\gamma_{ji} > 0$ for $i < j$. The encoding frequencies for these transitions should each be unique with $\gamma_{ij}=-\gamma_{ji}$ and can be chosen using base $B$ encoding.
\end{enumerate}

After the modulated Schr\"odinger equation Eq.~\eqref{eq:mod_schr} is solved at a sufficient number of $s$-points, a fast Fourier transform of Eq.~\eqref{eq:actual_abhra_linear_eq} is performed to yield the pathway class amplitudes associated with each frequency; the amplitude of a frequency $\gamma$ is the sum of all pathway amplitudes in the Hermitian pathway class whose pathways have frequency \mbox{$\gamma^{n(l_1,\dots,l_{n-1})}_{ba} = \gamma$}. To interpret the resulting mechanism, these frequencies must be translated into pathway classes (represented by a shortest pathway in the class). If the quantum system is sufficiently simple, the structure of the most significant pathway classes is often immediately evident from their frequencies, so translating frequencies can be done by inspection as in the numerical illustrations (Section \ref{sec:numerical_illustration}). Algorithm \ref{alg:freq_path_OHPE} is provided in Appendix \ref{apx:algo_FTM} to automate frequency translation in cases where manual inspection is infeasible. This algorithm takes the form of a pre-processing step constructing a map which, given a frequency, $\gamma$, produces a shortest pathway whose frequency is $\gamma$. After translation, the amplitudes extracted from Eq.~\eqref{eq:actual_abhra_linear_eq} are readily seen to correspond to pathway classes and we can analyze the resulting mechanism.

\subsection{Optimized Non-Hermitian Pathway Encoding}
\label{ssec:NHPE}

Here we will describe which matrix elements of the Hamiltonian should be modulated in NHPE for the efficient extraction of non-Hermitian pathway class amplitudes. Because this section deals with non-Hermitian encoding, the \mbox{$\ket i \!\to\! \ket j$} and \mbox{$\ket j \!\to\! \ket i$} transitions will be treated separately. As in Section \ref{ssec:OHPE}, the procedure enumerated below selects which matrix elements of the Hamiltonian should be modulated in NHPE for the purpose of solving Eq.~\eqref{eq:mod_schr} and performing a fast Fourier transform. The procedure for NHPE is similar to that of OHPE (Section \ref{ssec:OHPE}):
\begin{enumerate}
    \item Generate the \textit{directed} Hamiltonian graph with states as vertices and allowed transitions as arcs (directed edges) in both directions.
    \item Generate a spanning tree of the directed Hamiltonian graph. This can be done by starting with a spanning tree of the undirected graph and including only forward arcs in the directed graph's spanning tree.
    \item Encode transitions $\ket i \!\to\! \ket j$ corresponding to arcs \textbf{not} included in the spanning tree. The encoding frequencies for these transitions should each be unique with $\gamma_{ij}\neq-\gamma_{ji}$  and can be chosen using base $B$ encoding.
\end{enumerate}
If the spanning tree is chosen with only forward arcs (as described in the second step), the lower triangular portion of the dipole matrix is modulated as in OHPE with unique frequencies for non-spanning-tree transitions, but all transitions in the upper triangle of the Hamiltonian are modulated with additional unique frequencies.

Just as with Hermitian pathway classes, it is standard to represent non-Hermitian pathway classes by the lowest-order pathway in their class. In this paper, translation of frequencies is done by inspection, but Algorithm \ref{alg:freq_path_NHPE} is provided in Appendix \ref{apx:algo_FTM} to automate the process. Again, this translation allows the amplitudes extracted from Eq.~\eqref{eq:actual_abhra_linear_eq} to readily correspond to pathway classes thereby enabling the analysis of the resulting mechanism.

\subsection{Computational Cost of OHPE and NHPE}
The computational cost of Hamiltonian encoding comes from the range of possible pathway class frequencies. If $k$ transitions are encoded in base $B$ encoding then there are $B\times B^{k-1} = B^k$ possible pathway class frequencies in the encoding. Each extra frequency to be distinguished requires an additional $s$-point, so reducing the number of encoded transitions is the easiest way to reduce computational cost. In a $d$-dimensional system with ${r}$ allowed transitions, the original Hermitian encoding \cite{abhra_1} encodes all ${r}$ transitions. There are \mbox{$d-1$} edges in a spanning tree, so OHPE encodes \mbox{${r}-(d-1)$} transitions; hence, the computational cost of OHPE is $B^{r-d+1}$ numerical integrations of Eq.~{\eqref{eq:mod_schr}}. The general formula for the computational cost reduction provided by OHPE for a given $B$ is:
\begin{equation}
  B^{{r}}/B^{{r}-d+1} =  B^{d-1}.
\end{equation}
This formula does not depend on the number of edges in the graph and indicates that the larger the dimension of the system, the more significant the advantage provided by OHPE. In the same system, the original non-Hermitian encoding \cite{abhra_1} encoded $2{r}$ transitions (forward and backward transitions are treated separately) while NHPE encodes $2{r}-(d-1)$ and thus requires $B^{2{r}-d+1}$ numerical integrations of Eq.~\eqref{eq:mod_schr}. The improvement provided by NHPE for a given $B$ is:
\begin{equation}
    B^{2{r}}/B^{2{r}-d+1} = B^{d-1}.
\end{equation}
These improvement factors depend strongly on $B$ (and exponentially on $d$); the practical goal is to operate at the minimal self-validating base $B_{\min}+1$, where $B_{\min}$ exactly captures all Hermitian pathway classes with magnitude greater than $\epsilon$. Because $B_{\min}$ depends on the chosen amplitude magnitude threshold $\epsilon$, $B$ should be considered an independent variable when determining the computational cost of an encoding. Moreover, since $d$ and $r$ are parameters that characterize the specific quantum system,\footnote{In physical systems, $r$ may take any value from $d-1$ to $\binom{d}{2}$. In systems with only nearest-neighbor transitions, $r=d-1$; in $n$-qubit systems (with $d=2^n$), $r=n2^{n-1}$; in molecular systems with all transitions available, \mbox{$r=\binom{d}{2}=\frac{d(d-1)}{2}$.}} they should be considered constant factors. Additionally, while the improvement factors for NHPE and OHPE share the same formula, non-Hermitian encoding often has a much larger $B_{\min}$ for a given $\epsilon$ than Hermitian encoding due to the inclusion of backtracking information. Thus, the factor of improvement for NHPE is in practice larger than that of OHPE for the same control, though not enough to offset the inherent higher cost of NHPE. 

\section{Numerical Illustrations}\label{sec:numerical_illustration}

In this section, the new mechanism analysis procedures are applied to two quantum systems chosen for illustrative purposes; further works will examine particular types of physical systems (e.g.~the mechanism of quantum gate generation) in detail. First, OHPE and NHPE are applied to a three-level system. Then, OHPE is applied to an eight-level system modeled after three coupled qubits. To contrast these quantitative applications of encoding-based mechanism analysis, we will compare their results against the qualitative information that can be deduced from standard mechanistic references such as the systems' state population plots and the controls' frequency-time spectrograms.

Eq.~\eqref{eq:mod_schr} is solved numerically by approximating the control field $\varepsilon(t), t \in [0, T]$ as piecewise constant in time:
\begin{align} \label{eq:numerical_schr}
    U(T; s) &\approx \prod_{n=1}^{T/\Delta t} \exp \qty( -\frac{i}{\hbar} V_I(n\Delta t; s) \Delta t)  ,\nmberthis\\
    \intertext{where}
    V_I(n\Delta t; s) &=  -e^{iH_0n\Delta t/\hbar} \mu(s) e^{-iH_0n\Delta t/\hbar} \varepsilon(n\Delta t). \label{eq:numer_interaction_hamiltonian}
\end{align}
and the encoded dipole $\mu(s)$ is modulated as in Eq.~\eqref{eq:general_mod_mu}. Because the modulation given by OHPE yields a Hermitian modulated Hamiltonian, the matrix exponentials in Eq.~\eqref{eq:numerical_schr} may be efficiently calculated by diagonalizing the dipole $\mu(s)$ \mbox{\textit{a priori}}; in contrast, NHPE yields a non-Hermitian modulated Hamiltonian, so other techniques such as scaling and squaring \cite{exp} must be used to stably compute the matrix exponentials.

\subsection{Three-Level System}\label{ssec:3l-numil}

The first system considered is an arbitrary 3-level system with a Hamiltonian $H = H_0 - \mu\varepsilon(t)$ described in atomic units as:
\begin{equation} \label{eq:H0_3HF}
    H_0 = \mqty(\dmat[0]{0, 0.0082, 0.016})
\end{equation}
with
\begin{equation}\label{eq:Hc_3HF}
    \mu = \mqty(0 & 0.061 & -0.013 \\ 0.061 & 0 & 0.083 \\ -0.013 & 0.083 & 0).
\end{equation}
We will analyze an optimal control pulse produced by gradient ascent \cite{GRAPE} to perform population transfer from $\ket 1$ to $\ket 3$. Note that although \mbox{$\ket1\!\to\!\ket2$} and \mbox{$\ket2\!\to\!\ket3$} are nearly degenerate transitions in the field-free Hamiltonian, the dipole coupling elements break that symmetry to readily permit full control.

The population plot of Figure \ref{fig:3L_pop_graph} qualitatively suggests a fairly complicated mechanism---all three states have continued interactions throughout the duration of the pulse---but it is difficult to glean a more detailed description of the mechanism from this plot. Additionally, due to the short timescale of the control, a spectrogram of the control field fails to yield any useful information about the mechanism (and is therefore not pictured). The quantitative nature of encoding-based mechanism analysis draws a sharp contrast against the modest qualitative results above; using OHPE and NHPE, a quantitative and detailed description of the mechanism via pathways is presented below.

\begin{figure}[t]\captionsetup{font=small}
    \includegraphics[width=\linewidth]{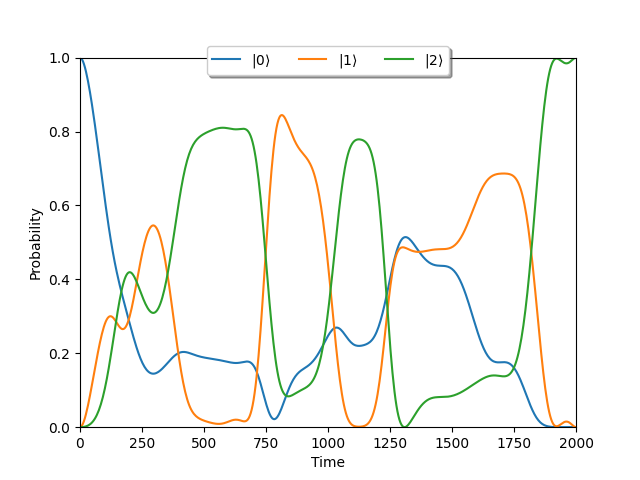}
    \caption{\justifying Populations of all states in the three-level system described by Eqs.~\eqref{eq:H0_3HF} and \eqref{eq:Hc_3HF} under the effect of the control performing a $\ket 1$ to $\ket{3}$ population transfer.}
    \label{fig:3L_pop_graph}
\end{figure}

\subsubsection{Hermitian Pathway Class Analysis}\label{ssec:3l-ohpe}

Here we analyze the mechanism underlying the controlled dynamics in Figure \ref{fig:3L_pop_graph} using OHPE. A graph representation of this system is provided in Figure \ref{subfig:3L_graph_H}. The three states and three transitions form a complete graph on three vertices (with three edges). In this graph, any two edges form a spanning tree, and the third edge is the \textit{only} transition modulated. As the choice of spanning tree is arbitrary, modulating any one of the three edges would reveal the same information. We choose the spanning tree that includes the \mbox{$\ket 1 \!\to\! \ket 2$} and \mbox{$\ket 2 \!\to\! \ket 3$} transitions as depicted in Figure \ref{subfig:3L_graph_H}, so the modulated transition is \mbox{$\ket 1 \!\to\! \ket3$.} With base $B=7$ encoding, the only modulation frequency is $\gamma_{31} = 1\gamma_0$ with $\gamma_0 = \frac{2\pi}8$ ($B$ manifests itself in the denominator of $\gamma_0$: 8 is the largest power of 2 no less than $B^{r-d+1}=7$).

\begin{figure}[t]\captionsetup{font=small}
    \centering
         \begin{subfigure}[b]{0.22\textwidth}
         \centering
         \includegraphics[width=\textwidth]{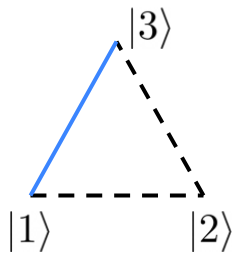}
         \caption{\justifying  The Hamiltonian graph of the 3-level system under OHPE.}
         \label{subfig:3L_graph_H}
     \end{subfigure}
     \hfill
     \begin{subfigure}[b]{0.22\textwidth}
         \centering
         \includegraphics[width=\textwidth]{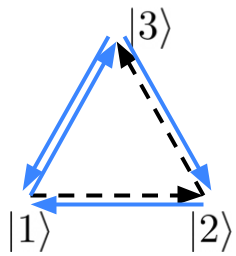}
         \caption{\justifying The directed Hamiltonian graph of the 3-level system under NHPE.}
         \label{subfig:3L_graph_NH}
     \end{subfigure}
    \hspace{-6pt}
    \vspace{2pt}
    
    \caption{\justifying The Hamiltonian graph for the three-level system under OHPE and NHPE. The spanning trees are depicted in dashed black. All other transitions (depicted in solid blue) are modulated.}
    \label{fig:3L_graph_HNH}
\end{figure}

Once the encoding is chosen, the $s$-sampling and fast Fourier transform can proceed as described in Section \ref{ssec:fourier}. When this is done, we are left with a list of Hermitian frequencies and their corresponding complex pathway class amplitudes. The next step is to convert the Hermitian frequencies with significant amplitudes into representative pathways for their corresponding Hermitian pathway classes. The representative pathways for these classes are constructed by inspection: due to the simplicity of this system, we can see that each net transition of the encoded edge corresponds to traversing a loop around the graph, and the number of traversals can be used to construct a representative pathway of the corresponding Hermitian pathway class. This idea is explored in more detail in Appendix \ref{apx:motiv_exam}. 

The most significant Hermitian pathway classes in the system's evolution are depicted in Table \ref{table:3lh}. Complex pathway amplitudes can be plotted as arrows in the complex plane, enabling the visualization of the constructive and destructive interference between pathway classes. The pathway amplitudes from Table \ref{table:3lh} are plotted in this manner in Figure \ref{fig:3l_arrow_h}.

\begin{table*}[t]\captionsetup{font=small}
\centering
\begin{tabular}{c l c c} 
 \hline
 \hline
 \addlinespace[2pt]$\gamma^{n(l_1,\dots,l_{n-1})}_{ba}/\gamma_0$ & Hermitian Pathway Class & Magnitude & Phase \\ [0.5ex] 
 \hline
 \addlinespace[2pt] 0 & $[\ket{1} \!\to\! \ket{2} \!\to\! \ket{3}]^\text{H}$ & 0.744 & 40$^\circ$  \\
  \addlinespace[2pt] -1 & $[\ket{1} \!\to\! \ket{2} \!\to\! \ket{3} \!\to\! \ket{1} \!\to\! \ket{2} \!\to\! \ket{3}]^\text{H}$ & 0.169 & 32$^\circ$ \\
 \addlinespace[2pt] 1 & $[\ket{1} \!\to\! \ket{3}]^\text{H}$ & 0.127 & 37$^\circ$ \\
 \addlinespace[2pt] -2 & {$[\ket{1} \!\to\! \ket{2} \!\to\! \ket{3} \!\to\! \ket{1} \!\to\! \ket{2} \!\to\! \ket{3} \!\to\! \ket{1} \!\to\! \ket{2} \!\to\! \ket{3}]^\text{H}$}  & 0.053 & 246$^\circ$ \\ 
  \addlinespace[2pt] 2 & $[\ket{1} \!\to\! \ket{3} \!\to\! \ket{2} \!\to\! \ket{1}\! \!\to\! \ket{3}]^\text{H}$ & 0.036 & 344$^\circ$ \\
 $\cdots$ & & $\cdots$ &$\cdots$ \\
 \hline
 Sum &  & 1.000 & 36$^\circ$ \\ 
 \hline
 \hline
\end{tabular}
\caption{\justifying Frequencies, corresponding Hermitian pathway classes, amplitude magnitudes, and amplitude phases of the three-level system driven by the control as seen in Figure \ref{fig:3l_arrow_h}. All significant pathway classes with magnitude greater than $\epsilon=0.01|U_{31}(T)|$ are listed; the bottom sum row includes additional (nonsignificant) pathway classes not shown and is equal to $U_{31}(T)$ exactly.}
\label{table:3lh}
\end{table*}

\begin{figure}\captionsetup{font=small}
    \includegraphics[width=\linewidth]{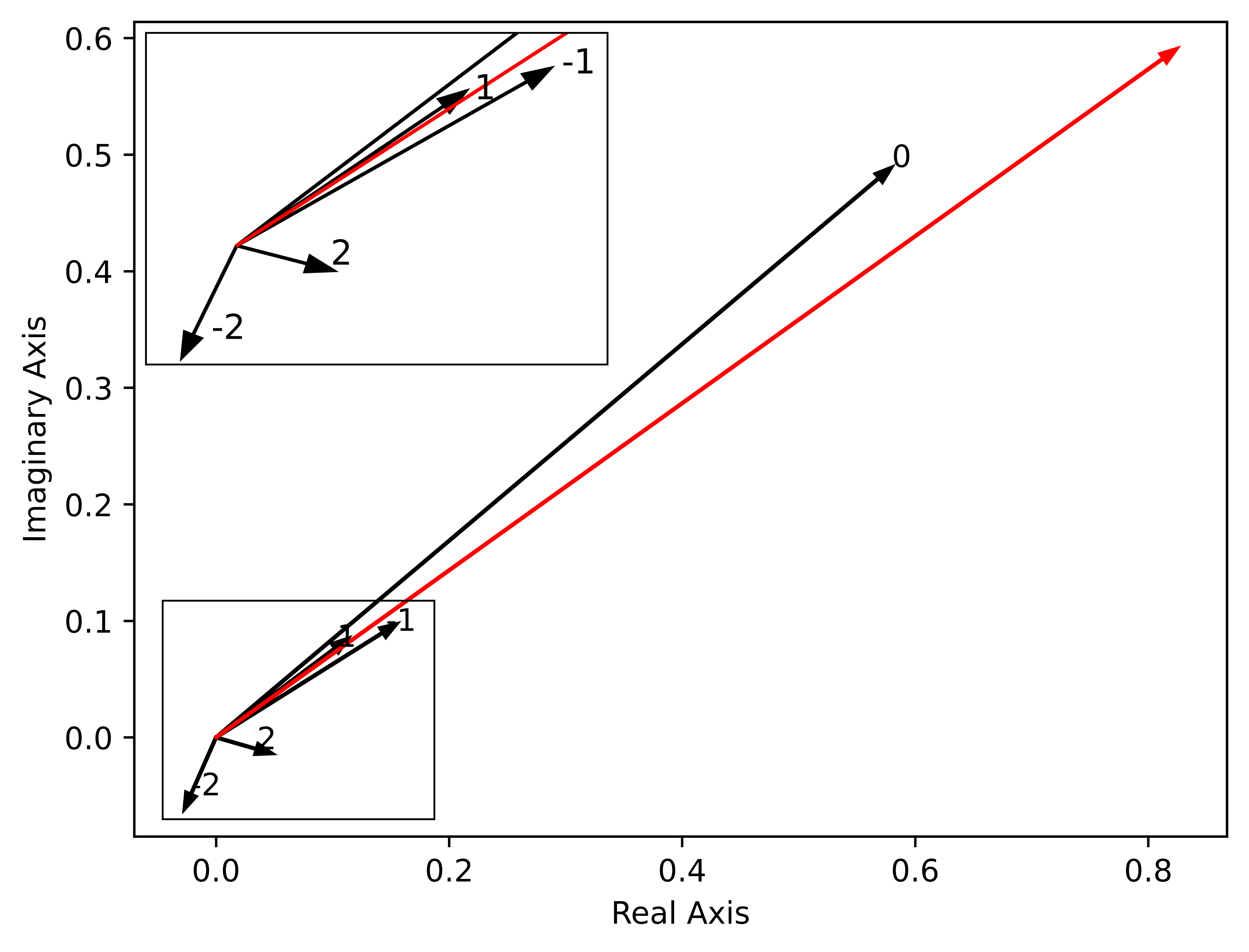}
    \caption{\justifying The Hermitian pathway class amplitudes of the controlled three-level system. Each amplitude is drawn as an arrow labeled with its associated frequency $\gamma^{{n(l_1,\dots,l_{n-1})}}_{ba}/\gamma_0$. All amplitudes with magnitude greater than \mbox{$\epsilon=0.01|U_{31}(T)|$} are displayed and labeled, and the inset expands the lower-left corner of the complex plane. The red arrow $U_{31}(T)$ is the sum of all pathway class amplitudes and has a magnitude close to $1$. The most significant pathways in the mechanism are expanded in Table \ref{table:3lh}.}
    \label{fig:3l_arrow_h}
\end{figure}

The Hermitian pathway class that uses only ladder-climbing transitions, \mbox{$[\ket{1} \!\to\! \ket{2} \!\to\! \ket{3}]^\text{H}$}, has largest pathway amplitude as seen in Table \ref{table:3lh}. This fact was not immediately obvious from the population plots in Figure \ref{fig:3L_pop_graph}. The two next most significant pathway classes are \mbox{$[\ket{1} \!\to\! \ket{2} \!\to\! \ket{3} \!\to\! \ket{1} \!\to\! \ket{2} \!\to\! \ket{3}]^\text{H}$} and \mbox{$[\ket 1 \!\to\! \ket 3]^\text{H}$}, which have comparable magnitudes to each other but are both much less significant than \mbox{$[\ket{1} \!\to\! \ket{2} \!\to\! \ket{3}]^\text{H}$} (likely due to the $\mu_{31}$ component being small compared to $\mu_{21}$ and $\mu_{32}$). After some point, as the complexity of Hermitian pathway classes increases, the magnitudes of their pathway amplitudes decrease. For example, the pathway classes with frequencies $-2\gamma_0$ and $2\gamma_0$ play a very small role in the mechanism. This, along with the lack of $3\gamma_0$ and $-3\gamma_0$ pathway classes with magnitude greater than $\epsilon=0.01|U_{31}(T)|$, shows that $B_{\min} = 5$ and thus $B = 7$ is the smallest self-validating odd base.

\subsubsection{Non-Hermitian Pathway Class Analysis}\label{ssec:3l-nhpe}

Here we analyze the mechanism underlying the controlled dynamics in Figure \ref{fig:3L_pop_graph} using NHPE. The directed Hamiltonian graph of this system is illustrated in Figure \ref{subfig:3L_graph_NH}. While the undirected Hamiltonian graph had 3 edges, the directed graph has 6 arcs. The same spanning tree that was used for the undirected Hamiltonian graph (Figure \ref{subfig:3L_graph_H}) is used for the directed graph. While OHPE only modulates one (undirected) transition in this system, NHPE requires four (directed) transitions to be modulated with different frequencies. A base of $B=16$ was chosen,\footnote{$\gamma_{31}=1\gamma_0$, $\gamma_{12}=16\gamma_0$, $\gamma_{13}=16^2\gamma_0$, and $\gamma_{23}=16^3\gamma_0$. $\gamma_0$ is equal to $\frac{2\pi}{2^{16}}$.} implying an encoding domain where each transition must be used less than 16 times. This base was chosen after first attempting to use $B=7$ and $B=10$, neither of which were self-validating with $\epsilon = 0.05|U_{31}(T)|$; performing a mechanism analysis with $B=16$ demonstrates that $B_{\min} = 13$, so $B=16$ is self-validating.

Modulating the Hamiltonian and extracting the non-Hermitian pathway class amplitudes yields a much larger list of significant pathway classes than in Hermitian encoding: 536 pathways with magnitude greater than \mbox{$\epsilon = 0.05|U_{31}(T)|$}. Table \ref{table:nh} details some significant pathway classes and illustrates the backtracking present in them. Figure \ref{fig:3l_arrow_nh} plots the complex pathway amplitudes of the most significant pathway classes. The non-Hermitian pathway class amplitudes depicted in Figure {\ref{fig:3l_arrow_nh}} fan out in all directions and have magnitudes much larger \mbox{than 1.} As NHPE gives a finer-grained mechanism analysis than OHPE, NHPE will provide a more qualitatively complex mechanism. This occurs because, due to the larger number of distinct pathway classes in finer-grained analyses, more destructive interference between pathway classes is necessary to reach the final amplitude (whose magnitude is at most $1$); this destructive interference between large-amplitude pathway classes of almost evenly-distributed phases yields the characteristic fan-out shape of Figure {\ref{fig:3l_arrow_nh}} and induces the large value of $B_{\min}=13$. Because non-Hermitian pathway classes are strict subsets of their respective Hermitian pathway classes, each pathway class amplitude in Section {\ref{ssec:3l-ohpe}} (i.e.~each entry of Table {\ref{table:3lh}} and each arrow of Figure {\ref{fig:3l_arrow_h}}) is the sum of a disjoint set of pathway class amplitudes in this analysis (i.e.~in Table {\ref{table:nh}} or Figure {\ref{fig:3l_arrow_nh}}).

The most significant non-Hermitian pathway class, with frequency $12304\gamma_0 = \gamma_{12} + 3\gamma_{23}$, magnitude $487$, and phase $285^\circ$, is \mbox{$[\ket1 \!\to\! \ket2 \!\to\! \ket1 \!$}\linebreak[4]\mbox{$\to\! \ket2\!\to\! \ket3 \!\to\! \ket2 \!\to\! \ket3 \!\to\! \ket2 \!\to\! \ket3 \!\to\! \ket2 \!\to\! \ket3]^{\text{NH}}$.} Including this class, the 22 most significant non-Hermitian pathway classes all come from the \mbox{$[\ket{1} \!\to\! \ket{2} \!\to\! \ket{3}]^\text{H}$} Hermitian pathway class (the most significant in Table \ref{table:3lh} and Figure {\ref{fig:3l_arrow_h}}); non-Hermitian pathway classes from other Hermitian pathway classes are also present in the mechanism, but they are generally less significant than the major non-Hermitian classes from the \mbox{$[\ket{1} \!\to\! \ket{2} \!\to\! \ket{3}]^\text{H}$} Hermitian class. The largest non-Hermitian pathway class not from \mbox{$[\ket{1} \!\to\! \ket{2} \!\to\! \ket{3}]^\text{H}$} is the 23rd most significant non-Hermitian class, with frequency \mbox{$1604\gamma_0 = \gamma_{31}+\gamma_{12}+4\gamma_{23}$}, magnitude 66, and phase $223^\circ$, is \mbox{$[\ket1 \!\to\! \ket2 \!\to\! \ket1 \!\to\! \ket3 \!\to\! \ket2 \!$}\linebreak[4]\mbox{$\to\! \ket3 \!\to\! \ket2 \!\to\! \ket3 \!\to\! \ket2 \!\to\! \ket3 \!\to\! \ket2 \!\to\! \ket3]^{\text{NH}}$}, which is from the \mbox{$[\ket 1 \!\to\! \ket 3]^{\text H}$} Hermitian pathway class.

\begin{figure}[t]\captionsetup{font=small}
    \includegraphics[width=\linewidth]{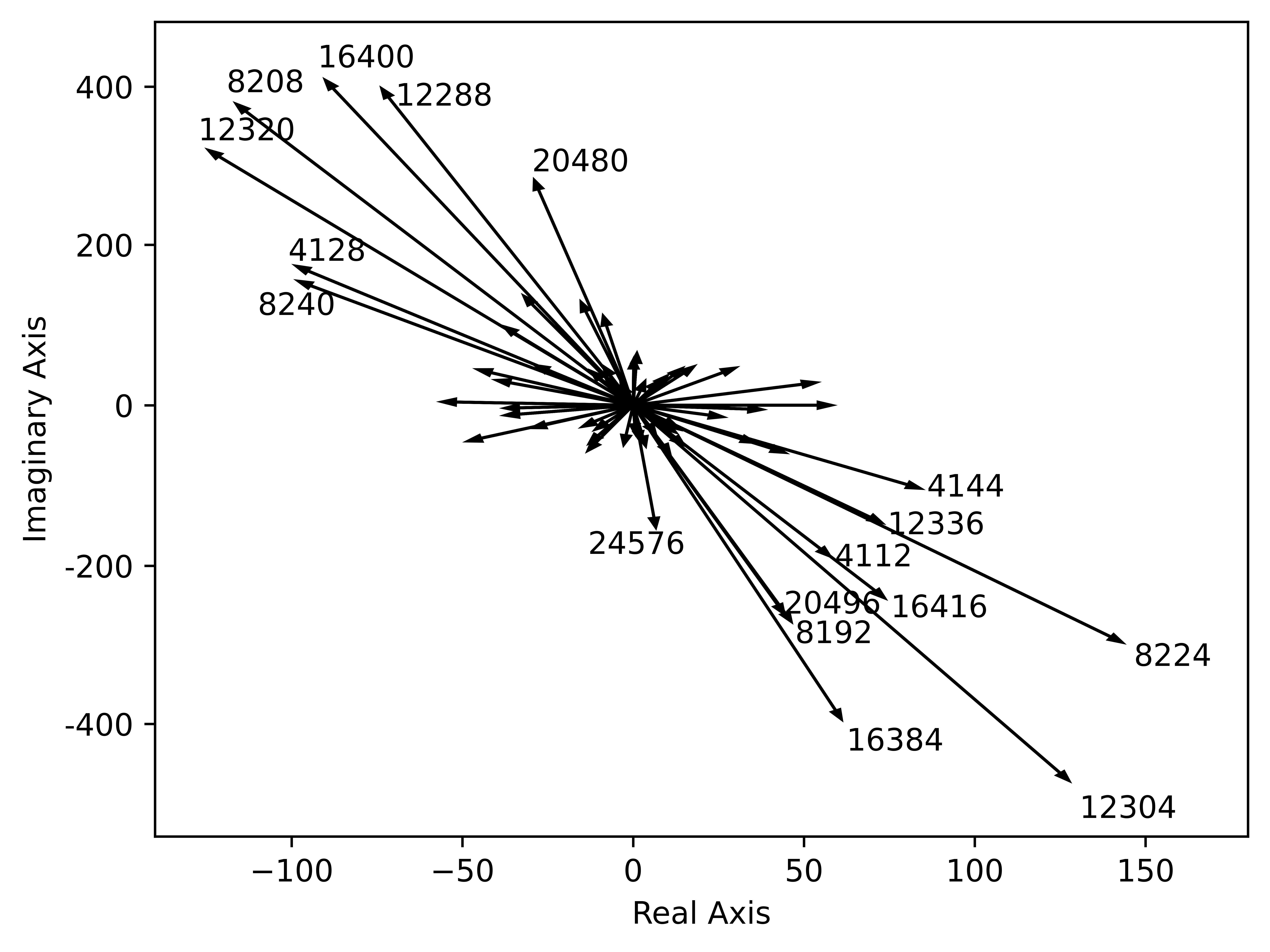}
        \caption{\justifying The non-Hermitian pathway class amplitudes of the three-level system driven by the control. Each amplitude is drawn as an arrow labeled with its associated frequency $\gamma^{{n(l_1,\dots,l_{n-1})}}_{ba}/\gamma_0$. There are too many significant pathways to legibly display in the figure, so only amplitudes with magnitude greater than $25$ are displayed and many significant pathways of lower magnitude are unlabeled. The sum of all pathway amplitudes $U_{31}(T)$ has magnitude close to $1$ and would not be visible if drawn on these axes. A few significant non-Hermitian pathway classes are expanded in Table \ref{table:nh}.}
    \label{fig:3l_arrow_nh}
\end{figure}

\begin{table*}\captionsetup{font=small}
\centering
\begin{tabular}{c c l c c} 
 \hline\hline
 \addlinespace[2pt]$\gamma^{n(\ldots)}_{ba}/\gamma_0$ & Decomposition & Non-Hermitian Pathway Class & Magnitude & Phase \\ [0.5ex] 
 \hline 
 \addlinespace[2pt]12288 & $3\gamma_{23}$ & {\small$[\ket{1} \!\to\! \ket{2} \!\to\! \underbrace{\ket{3} \!\to\! \ket{2} \!\to\! \ket{3} \!\to\! \ket{2} \!\to\! \ket{3} \!\to\! \ket{2} \!\to\! \ket{3}}]^\text{NH}$} & 403 & 100$^\circ$\\
 \addlinespace[2pt]8208 & $2\gamma_{23}+1\gamma_{12}$ & {\small$[\ket{1} \!\to\! \underbrace{\ket{2} \!\to\! \ket{1} \!\to\! \ket{2}} \!\to\! \underbrace{\ket{3} \!\to\! \ket{2} \!\to\! \ket{3} \!\to\! \ket{2} \!\to\! \ket{3}}]^\text{NH}$} & 394 & 107$^\circ$\\
 \addlinespace[2pt]4218 & $1\gamma_{23}+2\gamma_{12}$ & {\small$[\ket{1} \!\to\! \underbrace{\ket{2} \!\to\! \ket{1} \!\to\! \ket{2} \!\to\! \ket{1} \!\to\! \ket{2}} \!\to\! \underbrace{\ket{3} \!\to\! \ket{2} \!\to\! \ket{3}}]^\text{NH}$} & 200 & 119$^\circ$ \\
 \addlinespace[2pt]4112 & $1\gamma_{23}+1\gamma_{12}$ & {\small$[\ket{1} \!\to\! \underbrace{\ket{2} \!\to\! \ket{1} \!\to\! \ket{2}} \!\to\! \underbrace{\ket{3} \!\to\! \ket{2} \!\to\! \ket{3}}]^\text{NH}$} & 197 & 287$^\circ$ \\
 $\cdots$ & & & $\cdots$ & $\cdots$ \\
 \hline
 Sum & & & 1.000 & 36$^\circ$\\
 \hline \hline
\end{tabular}
\caption{\justifying Frequencies $\gamma^{n(l_1,\dots,l_{n-1})}_{ba}/\gamma_0$ in base $B=16$, basis frequency decompositions, corresponding non-Hermitian pathway classes, amplitude magnitudes, and amplitude phases of the controlled three-level system. Only four of many significant pathways are listed; these pathways were chosen for their short length. As in Table \ref{table:3lh}, the bottom sum row includes all pathway classes and is equal to $U_{31}(T)$ exactly. Parts of the pathway that correspond to backtracking have been under-braced, and removing the under-braced backtracking demonstrates that all four listed non-Hermitian classes contribute to the $[\ket{1} \!\to\! \ket{2} \!\to\! \ket{3}]^\text{H}$ Hermitian pathway class.}
\label{table:nh}
\end{table*}

NHPE enabled the mechanism of this three-level system with 6 allowed transitions (considering forwards and backwards transitions separately) to be fully extracted with only 4 modulated transitions. With 4 modulated transitions in base 16, Eq.~\eqref{eq:mod_schr} had to be solved $16^4$ times. Without NHPE, two more transitions would need to be modulated and Eq.~\eqref{eq:mod_schr} would need to have been solved at $16^6$ points, which is a factor of $16^2 =256$ improvement in computation time.

\subsection{Three-Qubit System}\label{ssec:3q-numil}

The second system analyzed here consists of three coupled qubits. The Hamiltonian is given by $H = H_0 + H_c(t)$. The field-free Hamiltonian takes the form:
\begin{align*} \label{eq:H0_3q_s}
    &H_0 = \omega_1 (I_z\! \otimes\! \mathbb{I}\! \otimes\! \mathbb{I}) \!+\! \omega_2 (\mathbb{I}\! \otimes\! I_z\! \otimes\! \mathbb{I}) \!+\! \omega_3 (\mathbb{I}\! \otimes\! \mathbb{I}\! \otimes\! I_z)\notag  \\
    &+\!J_{12} (I_z\! \otimes\! I_z\! \otimes\! \mathbb{I})\! +\! J_{13} (I_z\! \otimes\! \mathbb{I}\! \otimes\! I_z)\! +\! J_{23} (\mathbb{I}\! \otimes\! I_z\! \otimes\! I_z)\nmberthis
\end{align*}
where $I_z$ is the Pauli spin operator in the $z$ direction, $\mathbb I$ is the $2 \times 2$ identity matrix, $\omega_i$ are the resonance offset frequencies, and $J_{ji}$ are the spin couplings.\footnote{$\omega_1=2\pi \cdot 12039.6$, $\omega_2=2\pi \cdot -6855.5$, and $\omega_3=2\pi \cdot -12039.0$. The J-couplings are $J_{12}=2\pi \cdot 54$, $J_{13}=2\pi \cdot -1.3$, and $J_{23}=2\pi \cdot 35$.} The control Hamiltonian takes the form:
\begin{align} \label{eq:Hc_3q_s}
    H_c(t) &= \mu_1\varepsilon_1(t)  + \mu_2\varepsilon_2(t) \\
    \mu_1 &= I_x \otimes \mathbb I \otimes \mathbb I + \mathbb I \otimes I_x \otimes \mathbb I + \mathbb I \otimes \mathbb I \otimes I_x     \\
     \mu_2 &= I_y \otimes \mathbb I \otimes \mathbb I + \mathbb I \otimes I_y \otimes \mathbb I + \mathbb I \otimes \mathbb I \otimes I_y
\end{align}
where $I_x$ and $I_y$ are the Pauli spin operators in the $x$ and $y$ directions. $\mu_1$ and $\mu_2$ have the same structure, so when modulating the Hamiltonian, elements in the same position in both dipole matrices are multiplied by the same complex exponential. It is possible to encode the two dipole matrices separately to reveal which dipole's coupling matrix elements induce particular transitions in a pathway, but that direction is beyond the illustrative goals of this example.

The control used on this system performs an $X$ gate on the third qubit, transferring the population from $\ket{000}$ to $\ket{001}$. In this case, the state population plot (Figure \ref{fig:3q_pop_graph}) resulting from the two control pulses hints at certain aspects of the mechanism. In the middle of the population plot, there are clear peaks of states $\ket{100}$ and $\ket{101}$; from these two peaks, it is reasonable to hypothesize that the Hermitian pathway class \mbox{$[\ket{000} \!\to\! \ket{100} \!\to\! \ket{101} \!\to\! \ket{001}]^\text{H}$} might contribute significantly to the mechanism. In addition, there is a sharp uptick in the population of $\ket{011}$ near $t=0.0027$, suggesting that many pathways ending in a \mbox{$\ket{011} \!\to\! \ket{001}$} transition may be contributing. While it is possible to make such hypotheses, a truly quantitative understanding of the mechanism underlying the controlled dynamics is difficult to extract from population plots like Figure \ref{fig:3q_pop_graph} alone; furthermore, the spectrogram of the control field (not pictured) does not yield any mechanistic information that is not already evident in the population plot, so combining these standard mechanistic references still yields only a vague qualitative description of the mechanism. Beyond Hermitian pathway classes, this population plot also contains many short-lived dips and spikes in its curves. These features are most likely the result of Rabi flopping and can only be analyzed with non-Hermitian encoding. A non-Hermitian mechanism analysis of this system would be computationally expensive, so this simulation was not performed. 

\begin{figure}[hbt]\captionsetup{font=small}
    \includegraphics[width=\linewidth]{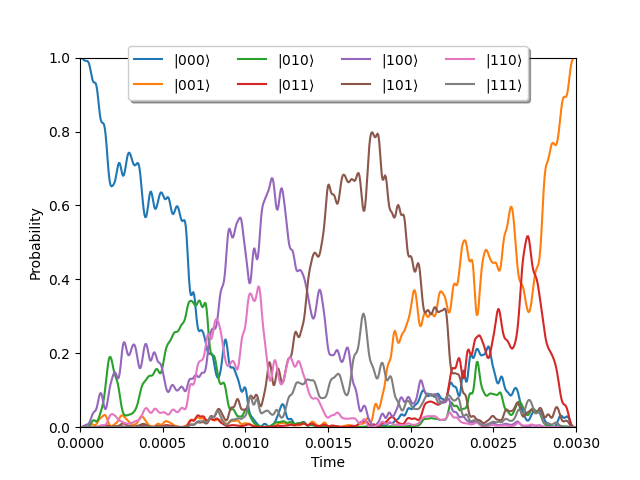}
    \caption{\justifying Populations of all states in the three-qubit system described by Eqs.~\eqref{eq:H0_3q_s} and \eqref{eq:Hc_3q_s} under the effect of a control performing an $X$ gate on the third qubit corresponding to transferring the population from $\ket{000}$ to $\ket{001}$, with $|U_{001,000}(T)|^2 = 0.99989$.}
    \label{fig:3q_pop_graph}
\end{figure}

\subsubsection{Hermitian Pathway Class Analysis}

The first step of OHPE is converting the system into a graph, which in this case is a cube; the Hamiltonian graph and spanning tree are presented in Figure \ref{fig:3q_system_graph}. The system has 12 allowed transitions, but using OHPE, modulating 5 transitions is enough to fully analyze Hermitian pathway classes. The Hamiltonian was encoded in base \mbox{$B=7$} Hermitian encoding;\footnote{$\gamma_{011,010}=1\gamma_0$, $\gamma_{110,010}=7\gamma_0$, $\gamma_{111,011}=7^2\gamma_0$, $\gamma_{101,100}=7^3\gamma_0$, and $\gamma_{111,110}=7^4\gamma_0$. $\gamma_0$ is equal to $\frac{2\pi}{2^{15}}$.} with \mbox{$\epsilon = 0.02|U_{001,000}(T)|$}, $B_{\min} = 5$ and $B=7$ is the smallest odd self-validating base.

\begin{figure}[hbt]\captionsetup{font=small}
    \begin{center}
        \includegraphics[width=0.3\textwidth]{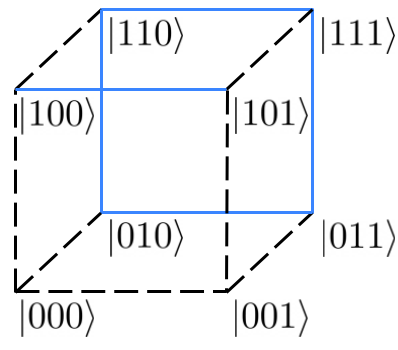}
    \end{center}
    \caption{\justifying The Hamiltonian graph of the three-qubit system is a cube. The vertices of the graph are labeled with their corresponding states. The spanning tree is depicted in dashed black. All other transitions (depicted in solid blue) are modulated.}
    \label{fig:3q_system_graph}
\end{figure}

The complex amplitudes of the Hermitian pathway classes are plotted in Figure \ref{fig:3q_arrow_h} and expanded in Table \ref{table:3qh}, and two example pathways are drawn in Figure \ref{fig:3q_example}. The population plot in Figure \ref{fig:3q_pop_graph} is suggestive that the \mbox{$[\ket{000} \!\to\! \ket{100} \!\to\! \ket{101} \!\to\! \ket{001}]^\text{H}$} pathway contributes significantly to the mechanism. This inference is consistent with the mechanism analysis: this pathway class (with frequency $343\gamma_0$) has the largest individual magnitude. In addition, as expected, most significant pathway classes end with either a \mbox{$\ket{101} \!\to\! \ket{001}$} transition or a \mbox{$\ket{011} \!\to\! \ket{001}$} transition. Beyond these insights based on the population plots, a Hamiltonian encoding mechanism analysis facilitates a quantitative understanding of mechanism which is much more detailed than a population plot could provide. For instance, we can see that despite the fact that the Hermitian pathway classes with frequencies $2352\gamma_0$ and $2359\gamma_0$ have pathways that only differ by one intermediate state, the former is much more significant than the latter; this distinction is not possible to predict from the population plot alone.  

\begin{table*}[ht]\captionsetup{font=small}
    \centering
    \begin{small}
        \begin{tabular}{c c l c c} 
     \hline\hline
     $\gamma^{n(\ldots)}_{ba}/\gamma_0$ & Decomposition & Hermitian Pathway Class & \!Magnitude\! & Phase \\ [0.5ex] 
     \hline
     \addlinespace[2pt] 343 & {$1\gamma_{101,100}$} & {\footnotesize$[\ket{000} \!\to\! \ket{100} \!\to\! \ket{101} \!\to\! \ket{001}]^\text{H}$} & 0.454 & 13$^\circ$\\
     \addlinespace[2pt] 2359 & {\!$1\gamma_{111,110}\! -\!1\gamma_{111,011}\! + \!1\gamma_{110,010}$\!}  & {\footnotesize$[\ket{000} \!\to\! \ket{010} \!\to\! \ket{110} \!\to\! \ket{111} \!\to\! \ket{011} \!\to\! \ket{001}]^\text{H}$} & 0.206 & 17$^\circ$\\
     \addlinespace[2pt] 350 & {$1\gamma_{101,100} + 1\gamma_{110,010}$} & {\footnotesize$[\ket{000} \!\to\! \ket{010} \!\to\! \ket{110} \!\to\! \ket{100} \!\to\! \ket{101} \!\to\! \ket{001}]^\text{H}$} & 0.134 & 346$^\circ$\\
     \addlinespace[2pt] 1 & {$1\gamma_{011,010}$} & {\footnotesize$[\ket{000} \!\to\! \ket{010} \!\to\! \ket{011} \!\to\! \ket{001}]^\text{H}$} & 0.089 & 299$^\circ$ \\
     \addlinespace[2pt] 2352 & {$1\gamma_{111,110} -1\gamma_{111,011}$} & {\footnotesize$[\ket{000} \!\to\! \ket{100} \!\to\! \ket{110} \!\to\! \ket{111} \!\to\! \ket{011} \!\to\! \ket{001}]^\text{H}$} & 0.065 & 108$^\circ$ \\   
      \addlinespace[2pt] 0 & 0 & {\footnotesize$[\ket{000} \!\to\! \ket{001}]^\text{H}$} & 0.063 & 338$^\circ$  \\ 
      $\cdots$ & & & $\cdots$ & $\cdots$ \\
     \hline
      Sum &  & &  1.000 & 0$^\circ$ \\
      \hline \hline
    \end{tabular}
    \end{small}
    \caption{\justifying Frequencies $\gamma^{n(l_1,\dots,l_{n-1})}_{ba}/\gamma_0$ in base $B=7$ encoding, basis frequency decompositions, corresponding pathways, pathway amplitude magnitudes, and pathway amplitude phases of the three-qubit system driven by the control from $\ket{000}$ to $\ket{001}$. All pathway classes with magnitude greater than \mbox{$\epsilon=0.05|U_{001,000}(T)|$} are listed. As in Tables \ref{table:3lh} and \ref{table:nh}, the bottom sum row includes all pathway classes and is equal to $U_{001,000}(T)$ exactly. The magnitude of the sum of all pathway amplitudes is $0.99994=|U_{001,000}(T)|$.} 
    \label{table:3qh}
\end{table*}

\begin{figure}[t]\captionsetup{font=small}
    \includegraphics[width=\linewidth]{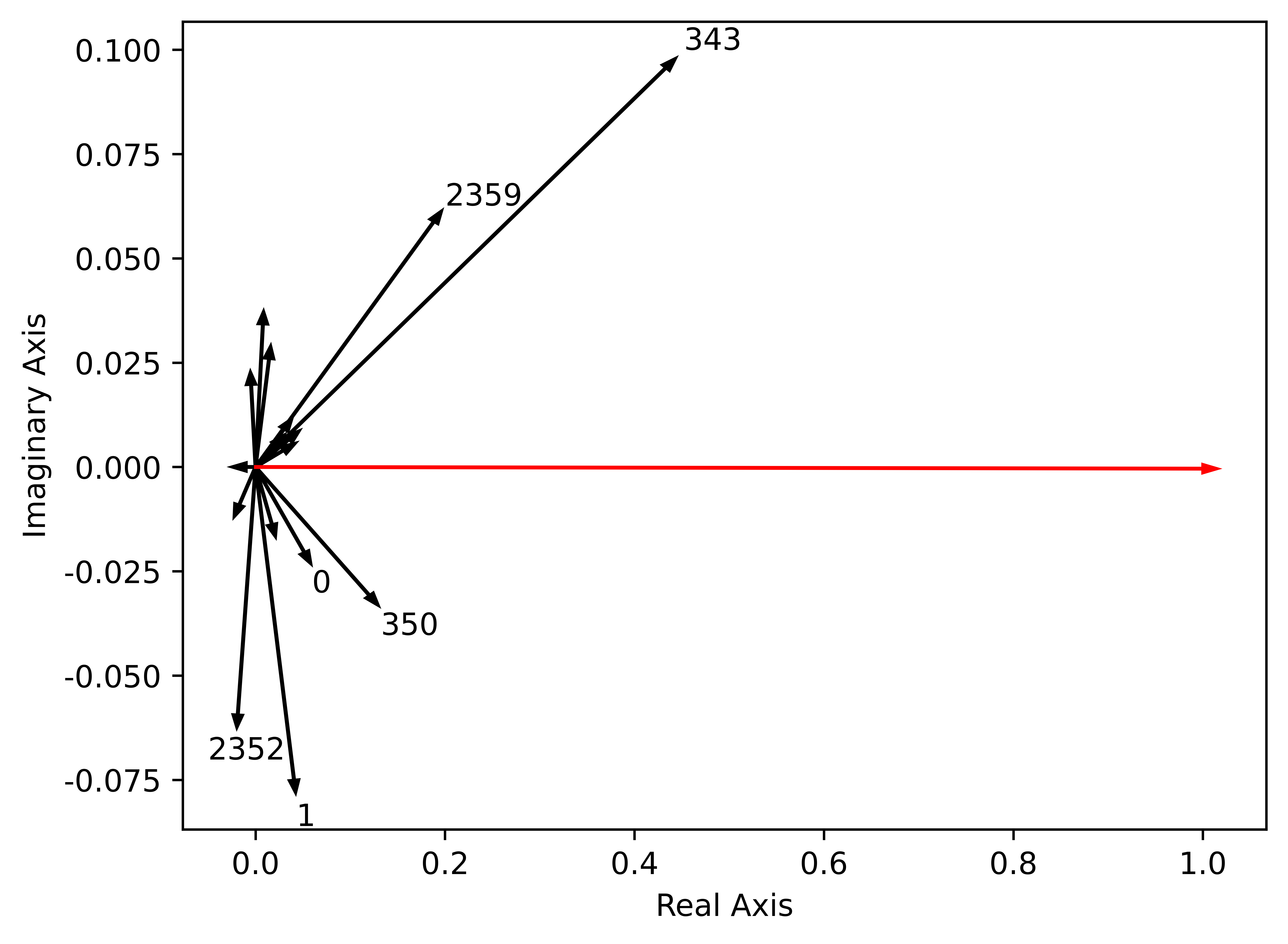}
    \caption{\justifying The Hermitian pathway class amplitudes of the controlled three-qubit system. Each amplitude is drawn as an arrow labeled with its associated frequency $\gamma^{{n(l_1,\dots,l_{n-1})}}_{ba}/\gamma_0$; pathway classes with magnitude greater than $0.05|U_{001,000}(T)|$ are labeled, and amplitudes with magnitude greater than \mbox{$\epsilon=0.02|U_{001,000}(T)|$} are displayed. The red arrow is the sum of all pathway class amplitudes and has a magnitude equal to \mbox{$|U_{001,000}(T)|=0.99994$}. Taking into account the different axis scales, there is strong constructive interference between the significant pathway classes. The most significant Hermitian pathway classes are expanded in Table \ref{table:3qh}.}
    \label{fig:3q_arrow_h}
\end{figure}

\begin{figure}[t]\captionsetup{font=small}
    \centering
         \begin{subfigure}[b]{0.22\textwidth}
         \centering
         \includegraphics[width=\textwidth]{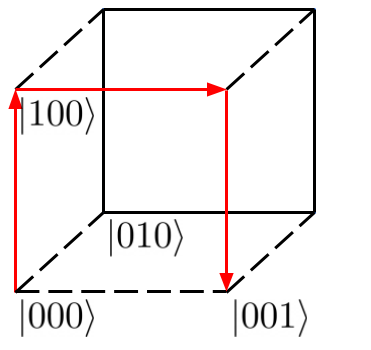}
         \caption{\justifying $\gamma^{n(l_1,\dots,l_{n-1})}_{ba}\!/\gamma_0 = 343 $}
         \label{subfig:p1}
     \end{subfigure}
     \hfill
     \begin{subfigure}[b]{0.22\textwidth}
         \centering
         \includegraphics[width=\textwidth]{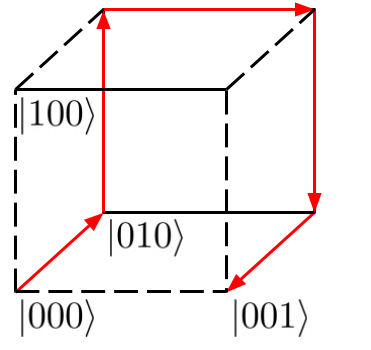}
         \caption{\justifying $\gamma^{n(l_1,\dots,l_{n-1})}_{ba}\!/\gamma_0 = 2359$}
         \label{subfig:p2}
     \end{subfigure}
    \hspace{-6pt}
    \vspace{2pt}
    \caption{\justifying The two most significant Hermitian pathway classes from $\ket{000}$ to $\ket{001}$ as seen in Table \ref{table:3qh} are illustrated on the graph of the three-qubit system. Spanning tree edges are marked in dashed black, encoded edges are in solid black, and the pathway is depicted in red arrows.}
    \label{fig:3q_example}
\end{figure}

In summary, although the population plot in Figure \ref{fig:3q_pop_graph} is sufficient to provide a qualitative picture of a portion of the mechanism, detailed quantitative information is not obtainable from the figure. In contrast, the encoding-based analysis enables a detailed description of the system's mechanism. In the three-level system of Section \ref{ssec:3l-numil}, the population plot provides very little mechanistic insight. In such cases, a full encoding-based mechanism analysis is the only viable method for obtaining a detailed picture of the constructive and destructive interference between the pathways involved in the dynamics.

\section{Conclusion}\label{sec:conclusion}
Hamiltonian encoding is a foundational tool for quantifying the mechanism underlying controlled quantum dynamics through the extraction of pathway amplitudes \cite{abhra_1, abhra_2, abhra_3, abhra_4}. Pathway classes encapsulate the contributions from many similar pathways into a single amplitude, and the interplay of pathway classes can be further interpreted through the lens of the constructive and destructive interference of their class amplitudes. The original encoding techniques \cite{abhra_1} for mechanism analysis were very computationally expensive and quickly became infeasible for systems with a large number of states. OHPE and NHPE provide an exponential improvement in the computational cost of mechanism analysis by exploiting topological properties of the set of allowed transitions in the quantum system. The body of this paper described these techniques and demonstrated that they offer substantial improvements in computational complexity without losing any mechanistic information. For a $d$-level system with ${r}$ allowed transitions and a given encoding base $B$, OHPE improves the cost of Hermitian encoding from $O(B^{{r}})$ to $O(B^{{r}-d+1})$ while NHPE improves the cost of non-Hermitian encoding from $O(B^{2{r}})$ to $O(B^{2{r}-d+1})$. Appendices \ref{apx:proof} and \ref{apx:cycles} prove the validity of OHPE by reinterpreting Hermitian pathway classes as describing sums of cycles on the Hamiltonian graph through the lens of algebraic topology, and Appendix \ref{apx:NHPE} extends these techniques to non-Hermitian encoding to prove the validity of NHPE.

The broad applicability of mechanism analysis may be further extended by generalizing the techniques of OHPE to more specialized analysis procedures. For example, Hamiltonian encoding has been successfully demonstrated in the laboratory \cite{rdc} where modulations are introduced via the control pulse instead of via the dipole moment, and future work will seek to use OHPE to design a more efficient procedure for deducing pathway amplitudes directly in the laboratory. Furthermore, Dyson-like expansions such as the Peano-Baker series \cite{Baake2011} appear often in solutions to first-order ordinary differential equations of the form $\dv{t}\vec{x}(t) = A(t) \vec x(t) + B(t) \vec u(t)$; it may be possible to use similar techniques to OHPE and NHPE to define and analyze the mechanism underlying such equations either by directly modulating $A(t)$ and decoding the resulting series solution. Another possible opportunity arises upon converting non-Hamiltonian linear differential equations into an equivalent Hamiltonian structure \cite{An2023, Jin2023} which again may be amenable to mechanism analysis.

Future works may also seek to improve unoptimized and unexplored aspects of OHPE and/or NHPE. Both of these algorithms currently take the spanning tree as input; it is possible that choosing spanning trees deliberately rather than arbitrarily may enable a lower base $B$ to be used. In addition, while NHPE provides an exponential reduction in computational cost over the original non-Hermitian encoding, its optimality has not been proven and more efficient non-Hermitian encoding schemes may exist. Finally, when systems are too large to be completely analyzed, it may be useful to encode and study subsets of the system; further research is needed to assess whether optimized encoding schemes like OHPE and NHPE are viable for such purposes.

One particularly promising application of these encoding techniques is in the analysis of qubit gates in the quantum information sciences. Due to the sparsity of qubit systems, OHPE and NHPE are particularly effective in optimizing their analysis. The mechanism underlying quantum gates can be analyzed to better understand why a given control has or does not have properties such as robustness to noise. Once the mechanistic patterns that are intrinsic to these properties are understood, the improved computation time of OHPE may make it possible to directly encourage these mechanistic patterns during the design of optimal controls. Future works will explore the outlined directions above.

\section*{Acknowledgements}
Michael Kasprzak acknowledges support from the Princeton Program in Plasma Science and Technology (PPST). Herschel Rabitz acknowledges partial support for theoretical research from the Department of Energy (DOE, Grant\# DE-FG-02ER15344). Tak-San Ho and Gaurav Bhole acknowledge partial support from the Department of Energy (DOE) STTR (Grant\# DE-SC0020618). Erez Abrams acknowledges the contributions of Edward Feng, Anthony Li, and Alexandra Volkova for drawing his attention to the relevance of spanning trees to mechanism analysis.

\section*{Appendices}
\appendix
\renewcommand\thefigure{\thesection\arabic{figure}}    
\setcounter{figure}{0}

\section{Motivating Example for Optimal Hermitian Pathway Encoding}\label{apx:motiv_exam}

Optimal Hermitian Pathway Encoding (OHPE) takes a simplified form in quantum systems whose Hamiltonian graphs have only one cycle. In such systems, the intuition behind OHPE is clear and the concepts learned can be generalized to arbitrary systems through the steps outlined in Section \ref{ssec:OHPE} and Appendix \ref{apx:proof}.

For illustration, a special case of OHPE is derived for a three-level quantum system identical to the system analyzed in Section \ref{ssec:3l-numil}. Consider the set of all pathways from $\ket{1}$ to $\ket{2}$ on the system depicted in Figure \ref{subfig:tri_1_graph}. The simplest possible pathway in this system is the direct transition \mbox{$\ket{1} \!\to\! \ket{2}$.} Higher-order pathways can traverse the graph in more complex ways; for instance, the pathway \mbox{$\ket{1} \!\to\! \ket{2} \!\to\! \ket{3} \!\to\! \ket{1} \!\to\! \ket{2}$} loops around the triangle counterclockwise and back to $\ket{1}$ before making the final transition to $\ket{2}$. In general, a pathway can loop around the graph counterclockwise multiple times yielding a pathway of the form \mbox{$(\ket{1} \!\to\! \ket{2} \!\to\! \ket{3}  \!\to\! )^k \ket{1} \!\to\! \ket{2}$} for $k \geq 0$, where the exponent represents how many times the pathway loops around the triangle. Additionally, a pathway can do this backward, looping around the graph in a clockwise direction, yielding pathways of the form \mbox{$(\ket{1} \!\to\! \ket{3} \!\to\! \ket{2} \!\to\! )^{-k} \ket{1} \!\to\! \ket{2}$} for $k \leq 0$ which have a single backtrack at the end. Each of these pathways is in a unique Hermitian pathway class, so we will now consider the set of Hermitian pathway classes defined by these pathways (one pathway class for each $k \in \ZZ$). For $k \geq 0$, these pathway classes are represented by the shortest pathway in their class, which is the pathway listed above, so the pathway class is represented as \mbox{$[(\ket{1} \!\to\! \ket{2} \!\to\! \ket{3} \!\to\! )^k \ket{1} \!\to\! \ket{2}]^\text{H}$.} For $k<0$, a simplification can be made: the $k<0$ pathways listed include backtracking, so the $k<0$ pathway classes are represented by their shortest pathways as \mbox{$[(\ket{1} \!\to\! \ket{3} \!\to\! \ket{2} \!\to\! )^{-k-1} \ket{1} \!\to\! \ket{3} \!\to\! \ket{2}]^\text{H}$.}

\begin{figure}[t]\captionsetup{font=small}
\centering

    \begin{subfigure}[b]{0.23\textwidth}
         \centering
         \includegraphics[width=\textwidth]{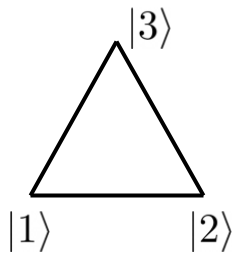}
         \caption{\justifying {Hamiltonian graph of the \phantom{(a)\,}three-level system}}
         \label{subfig:tri_1_graph}
     \end{subfigure}
     \hfill
     \begin{subfigure}[b]{0.23\textwidth}
         \centering
         \includegraphics[width=\textwidth]{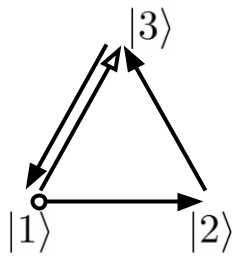}
         \caption{\justifying \mbox{$\!\ket{1} \!\to\! \ket{2} \!\to\! \ket{3} \!\to\! \ket{1} \!\to\! \ket{3}\!$} \phantom{text to align via wrap}}
     \end{subfigure}
    \hspace{-6pt}
    \vspace{2pt}
    \caption{\justifying The pathway \mbox{$\ket{1} \!\to\! \ket{2} \!\to\! \ket{3} \!\to\! \ket{1} \!\to\! \ket{3}$} on the three-level system starting with an open circle and ending with an open arrow. This is equivalent to the cycle \mbox{$\ket{1} \!\to\! \ket{2} \!\to\! \ket{3} \!\to\! \ket{1}$ }concatenated with \mbox{$\ket{1} \!\to\! \ket{3}$.} The Hermitian pathway class of this pathway can be represented by a simpler pathway: two final transitions cancel resulting in the pathway \mbox{$\ket{1} \!\to\! \ket{2} \!\to\! \ket{3}$.}}
\end{figure}

The key to optimizing Hermitian encoding is the understanding that we have just defined the set of \textit{all} possible Hermitian pathway classes in the system from $\ket{1}$ to $\ket{2}$. We can describe this set of pathway classes using a \textit{default pathway} and a cycle. We will choose \mbox{$\ket{1} \!\to\! \ket{2}$} as the default pathway in this system because it is the simplest pathway between the initial and final states. The default pathway on its own represents one Hermitian pathway class; all other Hermitian classes can be described by adding cycles to the default pathway. In this system, the only cycle is \mbox{$\ket{1} \!\to\! \ket{2} \!\to\! \ket{3} \!\to\! \ket{1}$,} which can be either traversed in the order described (arbitrarily defined to be forward) or in reverse (backward) some number of times before traversing the default pathway. We will define the set of pathway classes in this system in terms of an integer $k$, where if $k\geq 0$ then the pathway traverses the cycle forward $k$ times before the default pathway and if $k < 0$ then the pathway traverses the cycle backward $-k$ times before the default pathway. For example, the $k=1$ pathway class is \mbox{$[\ket{1} \!\to\! \ket{2} \!\to\! \ket{3} \!\to\! \ket{1} \!\to\! \ket{2}]^\text{H}$} and the $k=-1$ class is \mbox{$[\ket{1} \!\to\! \ket{3} \!\to\! \ket{2}]^\text{H}$.} Note that the pathway representing the $k=-1$ class does not explicitly contain the default pathway. In this way, the parameter $k$ entirely defines a pathway class and the only differences between pathway classes are described by their value of $k$.

Because Hermitian classes on this system are characterized by the number of times they traverse a single cycle, knowing the number of cycle traversals implies knowledge of the entire Hermitian pathway class. \textit{Critically, this means that only a single transition on the cycle needs to be modulated to extract full Hermitian pathway class information.} In this motivating example, take the \mbox{$\ket{2} \!\to\! \ket{3}$} transition and modulate it with some frequency $\gamma_0$; if a pathway has \mbox{$\gamma^{n(l_1,\dots,l_{n-1})}_{ba}=a\gamma_0$} then it belongs to the Hermitian pathway class with \mbox{$k=a$} in the set previously described. For example, a pathway with frequency $0\gamma_0$ corresponds to the class \mbox{$[\ket{1} \!\to\! \ket{2}]^\text{H}$,} a pathway with frequency $2\gamma_0$ corresponds to using the cycle twice counterclockwise as \mbox{$[\ket{1} \!\to\! \ket{2} \!\to\! \ket{3} \!\to\! \ket{1} \!\to\! \ket{2} \!\to\! \ket{3} \!\to\! \ket{1} \!\to\! \ket{2}]^\text{H}$,} and the frequency $-\gamma_0$ corresponds to using the cycle once clockwise as \mbox{$[\ket{1} \!\to\! \ket3 \!\to\! \ket{2}]^\text{H}$.} This demonstrates that in this system, all Hermitian pathway classes can be distinguished by only encoding a single transition.

The computational cost reduction provided by this optimization is very significant. If a base of 7 is used to define the frequencies, then if every transition is modulated, the frequency range has size $7^3\gamma_0 = 343\gamma_0$ so a minimum of 343 sample points are required to recover the amplitudes of the Fourier transform. Alternatively, if just the \mbox{$\ket{2} \!\to\! \ket{3}$} transition is modulated, then the frequency range is only $7^1 \gamma_0$ and therefore only 7 sample points are required to recover the amplitudes of the smaller Fourier transform. As shown above, the amplitudes of the smaller encoding are sufficient to reconstruct the amplitudes of the larger encoding. Therefore, this optimized encoding is 49 times faster than the unoptimized case. 

Thinking of pathways in terms of cycles in this way is what enables the computational cost reductions of OHPE. A full description, generalization, and proof of this concept are provided in Appendix \ref{apx:proof}, and the implications of this interpretation are discussed in Appendix \ref{apx:cycles}.

\section{Formulation and Proof of Optimal Hermitian Pathway Encoding}\label{apx:proof}
Optimal Hermitian Pathway Encoding (OHPE) exploits the backtracking-agnostic nature of Hermitian encoding by generating an encoding that forms a minimal complete basis for the set of cycles in a graph. We will prove that this novel encoding optimally reproduces the results of the original Hermitian encoding. The proof is detailed in three parts: (1) the graph theory and algebraic topology required to describe the problem in algebraic terms, (2) a description of the problem of optimizing Hermitian encoding in graph theoretic terms, and (3) a proof that OHPE is the optimal solution to that problem.
\subsection{Background and Definitions}\label{apx:b1}
First, we will describe the background necessary to follow the proof of OHPE. With the exception of Definitions 1 and 6, further information on the definitions and theorems of the following section can be found in standard graph theory (\cite{graph}) and algebraic topology (\cite{alg}) references.

\begin{definition}[Hamiltonian Graphs]
    Distinguishing pathway classes can be reduced to a problem in graph theory and algebraic topology. Consider Hamiltonians of the form $H(t) = H_0 - \mu\varepsilon(t)$. All transitions between states are encapsulated in the dipole matrix $\mu$. The Hamiltonian graph is constructed with one vertex for each eigenstate of $H_0$ by placing an edge between each pair of vertices that corresponds to an allowed transition in $\mu$. This is equivalent to constructing a graph using $\mu$ as an adjacency matrix (ignoring the numerical values and considering only whether they are nonzero). The Hamiltonian graph encapsulates the set of allowed transitions between the eigenstates. Consistent with standard graph theory notation, $V$ is the set of vertices and $E$ is the set of edges; $|V|=d$ is the dimension of the Hamiltonian and $|E| = {r}$ is the number of distinct transitions allowed by the system (the number of nonzero elements in the strictly upper triangular part of $\mu$). In general, if $\ket l$ is an eigenstate of $H_0$ then $v_l$ will refer to the corresponding vertex of the Hamiltonian graph.
\end{definition}

\begin{definition}[Orientation] In order to discern the direction of transitions on the Hamiltonian graph, we define an \textit{orientation} of the graph by giving every edge a \textit{forward} and \textit{backward} direction \cite{alg}. The orientation chosen is arbitrary and has no physical significance. For convenience, we will follow a standard procedure, the \textit{natural orientation}, when drawing graphs: instead of drawing undirected edges, we draw an oriented edge from $v_i$ to $v_ j$ if $\mu_{ji} \neq 0$ and $i < j$. An oriented edge can be traversed in either direction; the orientation exists only so that one direction is forward and the other direction is backward. By default, Hamiltonian graphs will be oriented using the natural orientation. 
\end{definition}

\begin{definition}[Directed Walk] A \textit{directed walk} on a graph is defined as a sequence of directed edges $w = \pns{e_1, \dots, e_{n-1}}$ such that there exists a sequence of vertices $\pns{v_1,\dots,v_n}$ where each $e_i$ is a directed edge from $v_i$ to $v_{i+1}$. A quantum pathway from $\ket{l_1}$ to $\ket{l_n}$ was defined in the introduction as a series of transitions \mbox{$\ket{l_{1}} \!\to\! \ket{l_{2}} \!\to\! \cdots \!\to\! \ket{l_{n}}$.} Because for each transition \mbox{$\ket{l_i} \!\to\! \ket{l_{i+1}}$} there must be an edge between $v_{l_{i}}$ and $v_{l_{i+1}}$, treating the sequence $\pns{v_{l_1}, \dots, v_{l_n}}$ as the vertex sequence of a walk leads to a natural bijection between quantum pathways and directed walks on the Hamiltonian graph. We will therefore use ``walk'' and ``pathway'' interchangeably hereafter. Because we use the natural orientation of the Hamiltonian graph, we can define a transition \mbox{$\ket{l_i} \!\to\! \ket{l_{i+1}}$} to be a \textit{forward transition} if \mbox{${l_i} <{l_{i+1}}$} and a \textit{backward transition} otherwise; this corresponds to a forward traversal or backward traversal of the corresponding oriented edge in the Hamiltonian graph. 
\end{definition}

\begin{definition}[Spanning Tree] A \textit{spanning tree} of a graph $G$ is a connected subgraph $T \subset G$ which has no cycles and includes all of $V$. A spanning tree of a graph with $|V|$ vertices has $|V| - 1$ edges. Note that any spanning tree includes a unique shortest directed walk from $v_S$ to $v_F$ for any $v_S,v_F \in V$.
\end{definition}

\begin{definition}[1-Chain] One useful property of a walk is the net number of times it traverses each oriented edge. When a transition occurs in the forward direction, the net number of transitions across that edge increments, and when a transition occurs in the backward direction, that number decrements. To formalize this, consider the $|E|$ dimensional vector space\footnote{Because $\ZZ$ is not a field, this is a free module over $\ZZ$, not a vector space; for convenience, we will call it a vector space regardless. In addition, it should be noted that free modules over $\ZZ$ can be viewed more simply as abelian groups $(\ZZ^n, +)$, but we will continue to refer to them as vector spaces so that we can more readily apply techniques from linear algebra.} $C_1 \cong \ZZ^{|E|}$ with basis vectors $\bcs{\vb e_1, \dots, \vb e_{|E|}}$. This is called the space of \textit{$1$-chains}. For every edge $e_i$, forward traversals of $e_i$ correspond to the vector $\vb e_i$, and backward traversals correspond to $-\vb e_i$. Given a directed walk $w$ from $v_S$ to $v_F$, define the 1-chain $w'$ of that walk to be the sum of the vectors corresponding to each transition. Every pathway has a unique 1-chain, but different pathways may have the same 1-chain. In addition, some 1-chains have no pathway; we will see that these 1-chains correspond to empty pathway classes, and eliminating them is a key result of the proof.
\end{definition}

\begin{definition}[Pathway Signature] By definition, the coefficients of a pathway's Hermitian frequency (Eq.~\eqref{eq:basis_freq}) are exactly the components of its 1-chain, which can be understood by considering the linear function $\Gamma: C_1 \to \RR$ defined by $\Gamma(\vb e_{ji}) = \gamma_{ji}$. We will define the \textit{full Hermitian pathway signature} (FHPS) of a pathway to be the list of coefficients in its basis frequency decomposition; this is the list of coefficients of its oriented edges, i.e.~its 1-chain. By definition, the set of all pathways that share an FHPS is called a Hermitian pathway class.
\end{definition}

\begin{definition}[Incidence Matrix] We can algebraically represent a graph with its incidence matrix $D$, which has $|V|$ rows and $|E|$ columns. Column $i$ represents oriented edge $e_i$ and row $j$ represents vertex $v_j$. If column $i$ represents the oriented edge from $v_a$ to $v_b$, then it has a $1$ in row $b$, a $-1$ in row $a$, and a $0$ everywhere else. The domain of this matrix is $C_1$; we will define the space of 0-dimensional chains or \textit{$0$-chains} $C_0 \cong \ZZ^{|V|}$ with basis vectors $\bcs{\vb v_1,\dots,\vb v_{|V|}}$ to be the codomain of this matrix. By the definition of $D$, if $e_i$ is the oriented edge from $v_a$ to $v_b$ then $D \vb e_i = \vb v_b - \vb v_a$. 
\end{definition}

\begin{definition}[Boundary] For a 1-chain $c$, its associated 0-chain $Dc$ is called the \textit{boundary} of $c$. The boundary of a pathway $p$ is defined to be the boundary of its 1-chain $p'$. This is physically meaningful as any pathway from $\ket i$ to $\ket j$ will have boundary $ \vb v_j - \vb v_i$. In addition, it should be noted that for any $v_S$ and $v_F$, every spanning tree $T$ has a unique 1-chain contained within it with boundary $\vb v_F -\vb v_S$, and all pathways in $T$ from $\ket S$ to $\ket F$ share that 1-chain.
\end{definition}

\begin{definition}[Cycle] The most natural definition of a cycle in this framework is that a cycle is a 1-chain whose boundary is 0. Explicitly, a cycle is a 1-chain in the kernel of $D$. This definition has a few key features:
\begin{enumerate}
    \item This definition of a cycle, which is the one used in algebraic topology \cite{alg}, is different from the graph theory definition of a cycle \cite{graph}. 
    \item For any pathway $p$ that starts and ends at the same vertex, its 1-chain $p'$ is a cycle. 
    \item Because the kernel of $D$ is a subspace of $C_1$, the set of cycles is a vector space whose dimension is the nullity of $D$.
\end{enumerate}
\end{definition}

With the equivalence between 1-chains and Hermitian pathway classes established, we no longer need to consider classes of quantum pathways (the equivalent but distinct formulation used in the main text). The entire problem has been translated to algebraic graph theory: Hamiltonians become directed graphs, pathways become directed walks, and Hermitian frequencies become 1-chains. Due to this equivalence, the problem can now be restated and proved using the tools of algebraic topology.

\subsection{Problem Description}\label{apx:problem_statement}

Here we define $P_{SF}$ to be the set of all pathways on a Hamiltonian graph $G = (V, E)$ from vertex $v_S$ to vertex $v_F$. Any pathway $p \in P_{SF}$ has a corresponding 1-chain, its \textit{full Hermitian pathway signature} denoted $\FHPS(p)\in C_1$, which can be calculated by placing a ``counter'' on each oriented edge of the Hamiltonian graph and counting the difference between the number of forward and backward transitions of that edge.

We note that modulating only some transitions of the Hamiltonian is equivalent to placing a counter only on some edges of the graph. Specifically, if $k<|E|$ edges are encoded, then the readout from those counters is an element of $\ZZ^k$ rather than of $C_1 \cong \ZZ^{|E|}$. 
The problem we aim to solve is as follows: {what is the minimal set of edges where counters must be placed so that if the readout of those counters is called the \textit{optimal Hermitian pathway signature} $\OHPS(p)$, then for all $p, q \in P_{SF}$, $\FHPS(p) = \FHPS(q)$ if and only if $\OHPS(p) = \OHPS(q)$?}

\subsection{Proof of Optimal Hermitian Pathway Encoding} \label{apx:proof_of_ohpe}
Theorems \ref{thm:2}--\ref{thm:5restrictioniso} have been proven before (e.g.~\cite{alg}), but we provide constructive proofs here that are more conducive to physical interpretation in the context of mechanism analysis. The main contribution in this context is summarized in Theorems \ref{thm:1}, \ref{thm:6FisoO}, and \ref{thm:7}, which provide a bridge between Hamiltonian encoding and algebraic topology, allowing these results to be applied in the optimization of mechanism analysis.

\begin{lemma} \label{thm:1}
Let $P'_{SF}$ be the set of 1-chains with boundary $\vb v_F - \vb v_S$. For any choice of ``default pathway'' $d \in P_{SF}$ with 1-chain $d'\equiv\FHPS(d) \in P'_{SF}$, the map $\phi_d(p') = p' - d'$ is a bijection $\phi_d : P'_{SF} \to \ker D$.
\end{lemma}
\begin{proof}
If $p' \in P'_{SF}$ then $Dp' = \vb v_F - \vb v_S$. Because $Dd' = \vb v_F - \vb v_S$ as well, we can immediately that see $D(p' - d') = 0$, so if $\phi_d$ has domain $P'_{SF}$ then its codomain is $\ker D$. Because $\phi_d^{-1}:\ker D \to P'_{SF}$ exists as $\inv\phi_d(c) = c + d'$, $\phi_d$ is a bijection.
\end{proof}

\begin{lemma}\label{thm:2}
The rank of $D$ is $|V| - 1$ and its nullity is $b_1 \equiv |E|-|V|+1$.
\end{lemma}
\begin{proof}
See proof in Theorems 4.3--4.5 of \cite{alg}.
\end{proof}

\begin{lemma} \label{thm:3cycleInTree} Fix a spanning tree $T \subset G$ (a spanning tree will be fixed for all future lemmas and theorems). Let $E \backslash T = \{e_{\tau_1}, \dots, e_{\tau_{b_1}}\}$ be the set of edges not in $T$. For every $e_{\tau_i}\in E\backslash T$, there exists a unique cycle $c_i$ in $G$ that includes only $e_{\tau_i}$ and edges in $T$, called the \textit{fundamental cycle} of $e_{\tau_i}$.
\end{lemma}
\begin{proof} Let $e_{\tau_i}$ be the directed edge from $v_a$ to $v_b$. All directed walks in $T$ from $v_b$ to $v_a$ have the same 1-chain: let $t_i'$ be that 1-chain. The 1-chain $\vb e_{\tau_i} + t_i'$ has boundary 0 and is unique because $t_i'$ is unique, proving the lemma. The fundamental cycle of $e_{\tau_i}$ is thus $c_i = \vb e_{\tau_i} + t_i'$.
\end{proof}

\begin{theorem} \label{thm:4}
Consider the set of fundamental cycles $\bcs{ c_i \stb  e_{\tau_i} \in E\backslash T}$. This set forms a basis for $\ker D$, called the \textit{fundamental cycle basis}.
\end{theorem}
\begin{proof}
See proof in Theorem 5.2 of \cite{alg}.
\end{proof}


\begin{definition} Let $W \cong \ZZ^{b_1}$ be spanned by basis vectors $\bcs{\vb w_1, \dots, \vb w_{b_1}}$ and consider the map $J : C_1 \to W$ defined by $J(\vb e_{\tau_i}) = \vb w_i$ for $e_{\tau_i} \in E \backslash T$ and $J(\vb e_i) = 0$ for $e_i \in T$. Note the following propositions:
\begin{enumerate}
        \item For any pathway $p$ with 1-chain $p'$, $J(p')$ corresponds to the set of readouts from the counters not on the spanning tree. Because there are $|V| - 1$ edges in the spanning tree, this means that $J(p')$ is the readout of $|E|-|V|+1 = b_1$ edges of the graph.
        \item The kernel of $J$ is the set of 1-chains contained entirely within $T$. 
        \item For all $v_S$ and $v_F$, there exists a unique 1-chain $d' \in \ker J \cap P'_{SF}$ because there exists a unique 1-chain in $T$ with boundary $\vb v_F-\vb v_S$.
\end{enumerate}
\end{definition}

\begin{theorem}\label{thm:5restrictioniso}
The restriction of $J$ to domain $\ker D$, denoted  $\nye J : \ker D \to W$, is an isomorphism.
\end{theorem} 

\begin{proof} First we note $\dim \ker D = \dim W$ by definition. We can show that $\nye J$ is an isomorphism by finding $\inv{\nye J}$, which we claim to be the fundamental cycle basis matrix, i.e.~the linear transformation defined by $\inv{\nye J} \vb w_i = c_i$. To demonstrate this, because $\bcs{c_i}$ form a basis for $\ker D$, it is sufficient to show that $\inv{\nye J} \nye J c_i = c_i$ for all $c_i$. First, we recall that $c_i$ was defined in Theorem 3 to be $\vb e_{\tau_i} + t_i'$ where $t_i' \in T$, i.e.~$t_i' \in \ker J$, so immediately we see that \mbox{$\nye J c_i = \vb w_i$.} By our expression for $\inv{\nye J}$, we see that $\inv{\nye J}  \nye J c_i = c_i$. \end{proof}

\begin{theorem} \label{thm:6FisoO} Fix a spanning tree $T$ in $G$. For any pathway $p \in P_{SF}$ with 1-chain \mbox{$\FHPS(p) \equiv p'$,} let its \textit{optimal Hermitian pathway signature} $\OHPS(p) \equiv Jp'$ be the readout of the counters on the edges of the graph not in $T$. If $d$ is the unique shortest pathway from $\ket S$ to $\ket F$ through $T$, then $\FHPS(p) = \phi_d^{-1}\inv{\nye J} \OHPS(p) $.
\end{theorem}

\begin{proof} Substitute $\FHPS(p) = p'$, $\OHPS(p) = Jp'$, and $\phi_d^{-1}(\inv{\nye J} J p')=\inv{\nye J} J p'+d'$, so we want to prove $p' = \inv{\nye J} J p' + d'$. Because $p'-d' \in \ker D$ and $d' \in \ker J$, we have:
\begin{align*}
    \inv{\nye J} J p' + d' &= \inv{\nye J} (J (p'-d') + Jd') + d'\\
    &= \inv{\nye J} (\nye J (p'-d') + 0) + d'\\
    &= p' - d' + d' = p'.
\end{align*}
\end{proof}

With the problem stated in Appendix \ref{apx:problem_statement} now solved, we can restate Theorem \ref{thm:6FisoO} in physical terms:

\begin{theorem} \label{thm:7} Encoding only the edges of the Hamiltonian graph not in a spanning tree enables the reconstruction of all Hermitian pathway classes with the minimum amount of modulated transitions $b_1 = {r} - d + 1$.
\end{theorem}

\begin{proof} The edges of the Hamiltonian graph modulated in an encoding correspond to the rows included in the matrix of $J$. Theorem 6 proves by construction that $\OHPS(p) = \OHPS(q)$ if and only if $\FHPS(p) = \FHPS(q)$, so by definition, modulating only the edges not in the spanning tree yields sufficient information to reconstruct all Hermitian pathway classes. In addition, because of the bijection between $P'_{SF}$ and $\ker D$ proven in Lemma 1, it is impossible to encode all Hermitian pathway class information without using at least $\dim \ker D = b_1$ edges, so this procedure is optimal in terms of its number of encoded edges. This result means that OHPE is optimal in terms of time and space complexity: when using base $B$ encoding, the original Hermitian encoding required $O(B^{r})$ numerical integrations of the modulated Schr\"odinger equation, but OHPE only requires $O(B^{b_1}) = O(B^{r-d+1})$ integrations, which is a factor of $O(B^{d-1})$ improvement in computation time.\end{proof}

This completes the proof. Due to the optimality of this procedure, we call the process of encoding only the transitions not contained in a spanning tree of the Hamiltonian \textit{Optimal Hermitian Pathway Encoding.}

\section{Encoding with Cycles and Reinterpreting Hermitian Pathway Classes}\label{apx:cycles}
\begingroup \binoppenalty=1000 \relpenalty=1000

The proof described in Appendix \ref{apx:proof} lends itself to a fundamental interpretation of what Hermitian pathway classes actually describe. Fix a spanning tree $ T$, an initial state $\ket S$, and a final state $\ket F$, and let the default pathway $d$ be the unique pathway from $\ket S$ to $\ket F$ through $T$. Appendix \ref{apx:proof_of_ohpe} proves that each non-tree edge $e_{\tau_i} \notin T$ corresponds to a unique cycle $c_i$ in the Hamiltonian graph (its fundamental cycle), and any Hermitian pathway class $p'$ from $\ket S$ to $\ket F$ can be decomposed as $c+d'$ with $c = \sum_{i} a_i c_i$. 

The optimal Hermitian pathway signature $\OHPS(p)$ is the net number of times each non-tree edge $e_{\tau_i}$ is traversed, i.e.~the information provided by OHPE. The proof reveals that the net number of times that a non-tree edge $e_{\tau_i}$ is traversed in $p$ is actually the coefficient $a_i$ corresponding to the number of times that the corresponding fundamental cycle $c_i$ is present in $c$. If we treat these cycles as modifications to the default pathway $d$, then $\OHPS(p)$ reveals the cycles which must be added to $d$ to represent a pathway in the Hermitian pathway class of $p$. This result shows us that, formally, Hermitian pathway classes have a deeper interpretation beyond discarding backtracking and time-sequencing information: Hermitian pathway classes describe fundamental and complete topological information about the presence of cycles in a pathway.

One reason that this interpretation is helpful is that it allows for any OHPS (and therefore any Hermitian frequency) to be easily converted into a representative pathway by hand. Figure \ref{fig:3q_apx} demonstrates this. Each labeled edge $e_{\tau_i}$ has the following fundamental cycle $c_i$: 
\begin{enumerate}
    \item[$c_1$:] Traverse the bottom face counterclockwise once ($\ket 3 \!\to\! \ket 4 \!\to\! \ket 2 \!\to\! \ket 1 \!\to\! \ket 3$).
    \item[$c_2$:] Traverse the top face clockwise once ($\ket 7 \!\to\! \ket 8 \!\discretionary{}{}{}\to\! \ket 6 \!\to\! \ket 5 \!\to\! \ket 7$).
    \item[$c_3$:] Traverse the left face clockwise once ($\ket3 \!\to\! \ket7\!\discretionary{}{}{} \to\! \ket5 \!\to\! \ket1 \!\to\! \ket3$).
    \item[$c_4$:] Traverse the right face counterclockwise once and the front face counterclockwise once ($\ket4 \!\to\! \ket8 \!\to\! \ket6 \!\to\! \ket5 \!\to\! \ket1 \!\to\! \ket2 \!\to\! \ket4$).
    \item[$c_5$:] Traverse the front face counterclockwise once ($\ket2 \!\to\! \ket6 \!\to\! \ket5 \!\to\! \ket1 \!\to\! \ket2$).
\end{enumerate}
Any cycle on the graph can be written as a linear combination of these basis cycles. For example, we can represent the right face counterclockwise as \mbox{$c_4-c_5$} and the back face clockwise as \mbox{$c_1-c_2-c_3+c_4$.} Note that for a cycle $c_i$ to be in the expansion $\sum_i a_i c_i + d'$ of a pathway's Hermitian class, the pathway need not traverse every edge in that cycle; it is common for at least one edge to not be traversed as multiple cycles are added together and backtracking transitions cancel (e.g.~in $c_4$).

\begin{figure}\captionsetup{font=small}
    \begin{center}
            \includegraphics[width=0.3\textwidth]{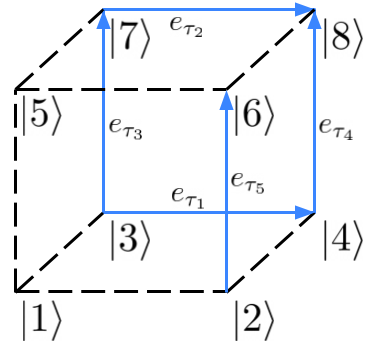}
    \end{center}
    \caption{\justifying The Hamiltonian graph of a cube system modeled after the 3-qubit system in \mbox{Figure \ref{fig:3q_system_graph}.} The vertices of the graph are labeled with their corresponding states. The spanning tree is depicted in dashed black. All other transitions (depicted in solid blue and labeled with $e_{\tau_i}$) are modulated.}
    \label{fig:3q_apx}
\end{figure}

A few examples showcase the ease of interpretation provided by this formulation. We will consider pathways from $\ket 1$ to $\ket5$, and for each example, we will find a representative pathway for the pathway class in question. First, the class with signature $(0,0,0,0,0)$ corresponds to the zero cycle, thus it is the Hermitian class of the default pathway \mbox{$\ket 1 \!\to\! \ket 5$.} A more complex pathway class may have signature $(0,0,0,1,-1)$, which corresponds to traversing the front face counterclockwise once and the top face counterclockwise once, so the pathway \mbox{$\ket1 \!\to\! \ket2 \!\to\! \ket6 \!\to\! \ket8 \!\to\! \ket7 \!\to\! \ket5$} modifies the default pathway to do exactly that. As the coefficients increase, translation stays simple: the pathway class with signature $(-1,0 ,0,1,-2)$ traverses the bottom face clockwise once, the front face counterclockwise once, and the top face counterclockwise twice; the pathway $\mathmbox{\ket1 \!\to\! \ket2} \!\discretionary{}{}{}\to\! \ket4 \!\to\! \ket3 \!\to\! \ket1 \!\to\! \ket2 \!\to\! \ket6 \!\to\! \ket8 \!\to\! \ket7 \!\to\! \ket5 \!\discretionary{}{}{}\to\! \ket6 \!\to\! \ket8 \!\to\! \ket7 \!\to\! \ket5$ satisfies this. 

Another reason why this formulation is useful is that it can be more constructive to think of Hermitian pathway classes in terms of their cycles rather than in terms of a representative pathway; doing so resolves the non-uniqueness of representative pathways which results from their inclusion of backtracking and time-sequencing information. Figure \ref{fig:ambiguity} demonstrates these issues with the traditional view: some pathway classes do not have a unique representative pathway, and in many of these cases, the shortest pathways in a class include some backtracking transitions which are not present in other pathways of the same class. The pathway class with signature $s_1=(0, 0, 1, -1, 0)$ has no unique representative pathway: the two shortest pathways in this class are  {$ \ket1 \!\to\!\ket 2 \!\to\!\ket 6 \!\to\!\ket 8 \!\to\!\ket 4 \!\to\!\ket 2 \!\to\!\ket 1  \!\to\!\ket 5$} and {$\ket1 \!\to\!\ket 5 \!\to\!\ket 6 \!\to\!\ket 8 \!\to\!\ket 4 \!\to\!\ket 2 \!\to\!\ket 6 \!\to\!\ket 5$.} Both of these pathways include backtracking as depicted in dashed green, and the specific backtracking transition used is distinct between the two pathways; these specific backtracking transitions are not present in every pathway in the class, so they are not fundamental to the pathway class and including them in the representative pathway may be misleading. A similar ambiguity occurs in the pathway class with signature $s_2 = (-1,  -1,0, 0, 0)$: the two shortest pathways in this class are {\begingroup \relpenalty=2000 \binoppenalty=2000 {$\ket1 \!\to\!\ket 2 \!\to\!\nobreak\ket 4 \!\discretionary{}{}{}\to\!\ket 3 \!\to\!\ket 1 \!\to\!\ket 5 \!\to\!\ket 7 \!\to\!\ket 3 \!\to\!\ket 1 \!\to\!\ket 5$} and {$\ket1 \!\to\!\ket 5 \!\discretionary{}{}{}\to\!\ket 7 \!\to\!\ket 3 \!\to\!\ket 1 \!\to\!\ket 2 \!\to\!\ket 4 \!\to\!\ket 3 \!\discretionary{}{}{}\to\!\ket 1 \!\to\!\ket 5$}\endgroup.} Here, the ambiguity arises from temporal differences in the pathways. Thinking of Hermitian pathway classes as fundamentally describing the cycles in a pathway resolves both of these ambiguities; with default pathway \mbox{$\ket 1 \!\to\!\ket 5$,} the Hermitian class with signature $s_1$ can now be interpreted as all pathways that traverse the right face once clockwise, while the Hermitian class with signature $s_2$ is all pathways that traverse the left face counterclockwise once and the bottom face clockwise once. This interpretation of pathway classes is more compact and more general than choosing a representative pathway, resolves the ambiguities inherent in representative pathway choice, and reflects the mathematical foundation of Hermitian encoding more directly.

\begin{figure}[t]\captionsetup{font=small}
    \centering
         \begin{subfigure}[b]{0.23\textwidth}
         \centering
         \includegraphics[width=\textwidth]{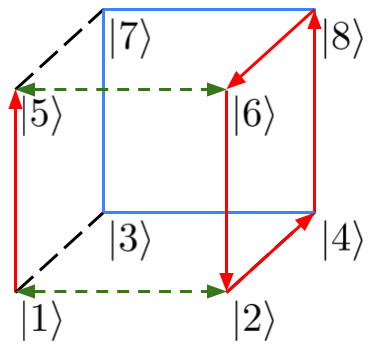}
         \caption{\justifying $s_1 = (0, 0, 1, -1, 0)$}
         \label{fig:s1}
     \end{subfigure}
     \hfill
     \begin{subfigure}[b]{0.23\textwidth}
         \centering
         \includegraphics[width=\textwidth]{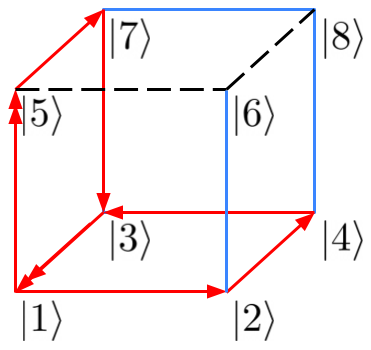}
         \caption{\justifying $s_2 = (-1, -1, 0,0,0)$}
         \label{fig:s2}
     \end{subfigure}
    \hspace{-6pt}
    \vspace{2pt}
    \caption{\justifying Two different Hermitian pathway classes from $\ket 1$ to $\ket 5$ are illustrated on the system in Figure \ref{fig:3q_apx}. Spanning tree edges are depicted in dashed black, encoded edges are depicted in solid blue, transitions that are present in every pathway in the pathway class are depicted with solid red arrows, and backtracking is depicted with dashed green arrows. Neither of these pathway classes has a unique shortest pathway: the shortest pathways of $s_1$ use different backtracking transitions to reach the right face of the cube, and the shortest pathways of $s_2$ differ by time-sequencing. Both of these ambiguities are resolved by considering pathway classes as representing sums of cycles rather than backtracking variations of a simplest representative pathway.}
    \label{fig:ambiguity}
\end{figure}
\endgroup

\section{Extending OHPE Techniques to Non-Hermitian Encoding}\label{apx:NHPE}

Optimal Hermitian Pathway Encoding works because discarding backtracking and time-sequencing information allows us to treat Hermitian pathway classes on a graph $G$ from $\ket S$ to $\ket F$ as part of a vector space, which makes sophisticated analysis and optimization possible. We wish to extend the techniques of OHPE to non-Hermitian encoding to decrease its computational cost; the encoding that we identify may not be optimal, but it will be faster than the original non-Hermitian encoding. To extend these techniques to non-Hermitian encoding, we must find some context in which non-Hermitian pathway classes lie in a vector space as well. To do so, we define the directed Hamiltonian graph $G^\to$: create a vertex for every state, and for each pair of states $\ket i$ and $\ket j$ with an allowed transition between them, construct the two arcs (directed edges) $v_i$ to $v_j$ and $v_j$ to $v_i$. Note that every arc in $G^\to$ corresponds to a unique nonzero off-diagonal element of the Hamiltonian and vice versa. The set of pathways on this directed graph is structured differently from the set of pathways on an oriented (undirected) graph. On an oriented graph, the pathways \mbox{$\ket 1 \!\to\! \ket 2$} and \mbox{$\ket 2 \!\to\! \ket 1$} each traverse the same edge but in different directions; on a directed graph, those pathways traverse two different arcs. Let the set of pathways from $S$ to $F$ on $G^\to$ be $P_{SF}^\to$; this set is manifestly in bijection with $P_{SF}$.

Consider the set of Hermitian pathway classes on $G^\to$ using the orientation defined by the directions of the arcs in $G^\to$ (rather than the natural orientation). These classes count the ``net number of traversals'' for every arc of $G^\to$, so they have twice as many counters as Hermitian pathway classes on $G$. Pathways can traverse each arc in only one direction, so the number of backward transitions is always 0. Because the net number of traversals of an arc in $G^\to$ is exactly the number of traversals of the corresponding edge in $G$ in one direction, two pathways are in the same Hermitian pathway class on $G^\to$ if and only if they are in the same non-Hermitian pathway class on $G$. As a result, we can directly apply OHPE to $G^\to$ to yield non-Hermitian pathway classes on $G$. This is done by constructing a spanning tree on $G^\to$ and encoding only the arcs not included in that spanning tree.

\section{Supplementary Algorithms}\label{apx:algo}
\begingroup
\hyphenpenalty=500 
This section details supplementary algorithms for use in the interpretation of results from OHPE and NHPE. Section \ref{apx:algo_FD} describes convenient algorithms for decomposing Hermitian or non-Hermitian frequencies encoded in base $B$ encoding into sums of basis frequencies so that the (net) number of transitions of each edge may be more easily read. 
Section \ref{apx:algo_FTM} describes algorithms for producing a map that translates frequencies into pathways. In practice, frequencies can be converted to representative pathways by hand. One way to do this is by inspecting the system graph and searching for a simple pathway that satisfies the (net) transition requirements imposed by the basis frequency decomposition; another way utilizes the cycle interpretation of pathway classes described in Appendix \ref{apx:cycles}. These algorithms can be useful in cases where manual conversion is infeasible, such as if there are too many significant pathway classes. 

\subsection{Frequency Decomposition}\label{apx:algo_FD}

The procedures of decomposing $\gamma^{{n(l_1,\dots,l_{n-1})}}_{ba}$ into basis frequencies are distinct for Hermitian and non-Hermitian encoding. To decompose a frequency in the non-Hermitian base $B$ encoding simply requires writing the integer $\gamma^{{n(l_1,\dots,l_{n-1})}}_{ba} / \gamma_0$ in base $B$. This is done explicitly in Algorithm \ref{alg:coef_nherm_freq}, which yields the coefficients $n_k$ used to construct a frequency $\gamma^{{n(l_1,\dots,l_{n-1})}}_{ba} = \sum_k n_k B^{k-1} \gamma_0$:

\begin{algorithm}[H]
\caption{\justifying Calculating Frequency Coefficients in Non-Hermitian Base $B$ Encoding}
\begin{algorithmic}
\State \texttt{int} $D \gets \gamma^{{n(l_1,\dots,l_{n-1})}}_{ba} / \gamma_0$
\State \texttt{list} $n \gets$ \texttt{list} of ${r}$ \texttt{int}s
\For{$k \gets 1, \dots , {r}$}
\State \texttt{int} $n[k] \gets D \mod B$
\State $D \gets (D - n[k]) / B$
\EndFor
\State \textbf{assert} $D = 0$
\State \Return $n$
\end{algorithmic}\label{alg:coef_nherm_freq}
\end{algorithm}

The procedure for decomposing $\gamma^{n(l_1,\dots,l_{n-1})}_{ba}$ into basis frequencies is less straightforward for Hermitian pathway classes. Decomposing a frequency in Hermitian base $B$ encoding requires accounting for the fact that the coefficients used are in $\Bqty{-{m_0}, \dots, 0 \dots, {m_0}}$ with ${m_0} = (B-1)/2$ instead of in $\Bqty{0, \dots, B-1}$. Algorithm \ref{alg:coef_herm_freq} yields the coefficients $n_k$ used to construct a frequency $\gamma^{{n(l_1,\dots,l_{n-1})}}_{ba} = \sum_k n_k B^{k-1} \gamma_0$ with odd $B$:
\begin{algorithm}[H]
\caption{\justifying Calculating Frequency Coefficients in Hermitian Base $B$ Encoding}
\begin{algorithmic}
\State \texttt{int} $D \gets {B}^{r} + \gamma^{{n(l_1,\dots,l_{n-1})}}_{ba} / \gamma_0$
\State \texttt{list} $n \gets$ \texttt{list} of ${r}$ \texttt{int}s
\For{$k \gets 1, \dots, {r}$}
\State \texttt{int} $d \gets D \mod B$
\If{$d < (B+1)/2$}
\State \texttt{int} $n[k] \gets d$
\Else 
\State \texttt{int} $n[k] \gets d - B$
\EndIf
\State $D \gets (D - n[k] ) / B$
\EndFor
\State \textbf{assert} $D = 1$
\State \Return $n$
\end{algorithmic}\label{alg:coef_herm_freq}
\end{algorithm}

In the Hermitian case, these coefficients correspond to the net amount of times each transition was used: if $\gamma_k$ is used $n_k$ times then the net number of times that transition $k$ was traversed is $n_k$.

\subsection{Frequency Translation Maps}\label{apx:algo_FTM}

The purpose of Algorithms \ref{alg:freq_path_OHPE} and \ref{alg:freq_path_NHPE} is to produce a \textit{translation map} given a Hamiltonian graph $G$, an initial state $\ket S$, a spanning tree $T$, and a length threshold $l$. A translation map is a function which, given a final state and a pathway class frequency, produces a low-order pathway in that pathway class. Both of these algorithms can be summarized succinctly as follows: iterate over all pathways of \mbox{length $\leq l$} and calculate their frequencies, then save a map from frequencies to a shortest pathway with that frequency. In practice, pathway class frequencies may often be translated to representative pathways by hand and no such iteration is needed, but these algorithms can be useful in cases where there are too many significant pathway classes for manual translation to be feasible.

Algorithm \ref{alg:freq_path_OHPE} utilizes the function \mbox{\textsc{PathwayFreqH}$(G, T, p)$} which returns the frequency of pathway $p$ under the given Hermitian encoding, while Algorithm \ref{alg:freq_path_NHPE} utilizes \mbox{\textsc{PathwayFreqNH}$(G, T, p)$} which returns the frequency of pathway $p$ under the given non-Hermitian encoding. In addition, both algorithms use a function to iterate over all pathways of length $\leq l$: \mbox{\textsc{IterPathwaysH}$(G, l)$} iterates over all pathways of \mbox{length $\leq l$} on $G$ using breadth-first search but does not allow for Rabi-flopping (transitions of the form \mbox{$\ket a \!\to\! \ket b \!\to\! \ket a$}), while \mbox{\textsc{IterPathwaysNH}$(G, l)$} iterates over all pathways of \mbox{length $\leq l$} using breadth-first-search without the Rabi-flopping restriction. The frequency translation map is constructed for Hermitian encoding as follows:

\begin{breakablealgorithm}
\caption{\justifying Constructing a Frequency Translation Map for OHPE}

\begin{algorithmic}
\State \texttt{map} $M \gets$ \texttt{map} from \texttt{(state, float)} to \texttt{pathway}

\ForAll{\texttt{pathway} $p \in $ \textsc{IterPathwaysH}$(G, l)$}
\State \texttt{state} $\ket{F} \gets $ final state of $p$
\State \texttt{float} $\gamma \gets$ \textsc{PathwayFreqH}$(G, T, p)$
\If{$(\ket F, \gamma)\notin M$}
\State $M$.add($(\ket F, \gamma) \mapsto p$)
\EndIf
\EndFor
\State \Return $M$
\end{algorithmic}
\label{alg:freq_path_OHPE}
\end{breakablealgorithm}

\noindent  The frequency translation map is constructed for non-Hermitian encoding similarly:

\begin{algorithm}[H]
\caption{\justifying Constructing a Frequency Translation Map for HHPE}

\begin{algorithmic}
\State \texttt{map} $M \gets$ \texttt{map} from \texttt{(state, float)} to \texttt{pathway}

\ForAll{\texttt{pathway} $p \in $ \textsc{IterPathwaysNH}$(G, l)$}
\State \texttt{state} $\ket{F} \gets $ final state of $p$
\State \texttt{float} $\gamma \gets$ \textsc{PathwayFreqNH}$(G, T, p)$
\If{$(\ket F, \gamma)\notin M$}
\State $M$.add($(\ket F, \gamma) \mapsto p$)
\EndIf
\EndFor
\State \Return $M$
\end{algorithmic}\label{alg:freq_path_NHPE}
\end{algorithm}

When interpreting mechanism, we consider only the most significant pathway classes. As such, each algorithm's length threshold should at first be set small and then increased until all significant pathway class frequencies (with amplitude greater than $\epsilon$) in an analysis have been translated into representative pathways.
\endgroup

\bibliographystyle{quantum}
\bibliography{sources}

\end{document}